\documentclass[accepted]{uai2025} 

\newcommand{\removed}[1]{}
\usepackage{times}
\usepackage{soul}
\usepackage{url}
\usepackage[utf8]{inputenc}
\usepackage{graphicx}
\usepackage{amsmath}
\usepackage{amsthm}
\usepackage{booktabs}
\usepackage{algorithm}
\usepackage{algorithmic}
\usepackage{stackengine}
\def\defeq{\mathrel{\ensurestackMath{\stackon[1pt]{=}{\scriptscriptstyle\Delta}}}}

\usepackage{algorithm}
\usepackage{algorithmic}

\usepackage{times}

\usepackage{url}

\usepackage{amsmath}

\usepackage{amssymb}
\usepackage{mathtools}
\usepackage{amsthm}

\usepackage{algorithmic}

\usepackage{lscape}
\usepackage[capitalize,noabbrev]{cleveref}

\theoremstyle{plain}

\usepackage[textsize=tiny]{todonotes}
\usepackage{multirow}

\usepackage{ascmac}
\usepackage{float}
\usepackage{perpage}
\MakeSorted{figure}
\MakeSorted{table}

\usepackage{url}
\usepackage{natbib}
\usepackage{chapterbib}

\usepackage{color}
\usepackage{tikz}
\tikzset{%
mynode/.style={circle,minimum width=.5ex, fill=none,draw}, 
myfillnode/.style={circle,minimum width=.5ex, fill=lightgray,draw}, 
}
\usepackage{amssymb}
\usepackage{natbib}

\newcommand{\indep}{\perp \!\!\! \perp}
\usepackage{amsmath}               
\usepackage{lscape}
\usepackage{algorithm}
\usepackage{color}
\usepackage{tikz}
\usepackage{amsmath,amsthm}
\newtheorem{theorem}{Theorem}
\newtheorem{definition}{Definition}
\newtheorem{assumption}{Assumption}
\newtheorem{lemma}{Lemma}

\usepackage{multirow}
\usepackage{comment}
\usepackage{here}
\allowdisplaybreaks[4]
\usepackage{caption}
\usepackage{bbding}
\usepackage{arydshln}
\usepackage{afterpage}

\usepackage{mathrsfs}

\usepackage{soul}

\usepackage{amsmath}
\usepackage{amssymb}
\usepackage{mathtools}
\usepackage{amsthm}
                 
\usepackage[american]{babel}

\usepackage{natbib} 
\usepackage{mathtools} 
\usepackage{booktabs} 
\usepackage{tikz} 

\title{Moments of Causal Effects}

\author[1]{\href{mailto:<Yuta.Kawakami@mbzuai.ac.ae}{Yuta Kawakami}}
\author[1]{\href{mailto:<Jin.Tian@mbzuai.ac.ae}{Jin Tian}}
\affil[1]{%
Mohamed bin Zayed University of Artificial Intelligence, UAE
}

\begin{document}
\maketitle

\begin{abstract}
The moments of random variables are fundamental statistical measures for characterizing the shape of a probability distribution, encompassing metrics such as mean, variance, skewness, and kurtosis. 
Additionally, the product moments, including covariance and correlation, reveal the relationships between multiple random variables.
On the other hand, the primary focus of causal inference is the evaluation of causal effects, which are defined as the difference between two potential outcomes.
While traditional causal effect assessment focuses on the average causal effect, this work  provides definitions, identification theorems, and bounds for moments and product moments of causal effects to analyze their distribution and relationships.
We conduct experiments to illustrate the estimation of the moments of causal effects from finite samples and demonstrate their practical application using a real-world medical dataset.
\end{abstract}


\section{Introduction}

The \emph{moments} of random variables have been fundamental statistical measures since their introduction by Pafnuti Lvovich Chebyshev in the mid-nineteenth century \citep{Mackey1980}.
The $m$-th moment of a random variable $Y$ is defined as $\mathbb{E}[Y^m]$, and the $m$-th central moment is defined as $\mathbb{E}[(Y-\mathbb{E}[Y])^m]$.
These moments characterize the shape of a random variable’s probability distribution, encompassing measures such as mean, variance, skewness, and kurtosis \citep{Pearson1896,Joanes1998,Cramer1999,Doane2011,Hippel2005,Westfall2014}.
The (central) moments of random variables also play a fundamental role in various machine learning techniques \citep{Bishop2006,Hastie2009,Murphy2022}.

On the other hand, the primary focus of causal inference is the evaluation of causal effects  $Y_1 - Y_0$, where $Y_x$ denotes the potential outcome under treatment $X = x$, rather than a single random variable $Y$.
Traditionally,  to assess causal effects, researchers  estimate the average causal effect (ACE), i.e., $\mathbb{E}[Y_1 - Y_0]$, which represents the first moment of causal effects \citep{Neyman1923, Rubin1978, Holland1986, Balke1997, Robins1999}.

Recently, there has been increasing interest in exploring aspects of causal effects beyond their average, particularly in the distributional properties of causal effects \citep{Ju2010, Wiedermann2022, Lin2023, Kennedy2023b, Post2023}.
{The shape of the distribution of causal effects uncovers causal effect heterogeneity, which is an actively researched topic in the field of statistics, causal inference, and machine learning \citep{Athey2016, Shalit2017, Athey2019, Kunzel2019, Wager2018, Singh2023, Kawakami2024b}.
Causal effect heterogeneity refers to the variation in causal effects across individuals or subgroups within a population.
Existing works on causal effect heterogeneity mainly examine the conditional average causal effects (CACE), i.e., $\mathbb{E}[Y_1-Y_0|W=w]$, based on subjects’ covariates $W$.
However, CACE captures only the heterogeneity across subpopulations specified by observed covariates $W$, not the heterogeneity across individuals.
In contrast, the shape of the distribution of causal effects reveals the heterogeneity of causal effects across individuals and provides complementary information to CACE.}

Our objective is to address 
the following causal question:
\begin{center}
(\textbf{Question 1}).
``{\it
How are causal effects distributed?
}"
\end{center}
We approach this question by studying the moments of causal effects $\mathbb{E}\Big[(Y_1-Y_0)^m\Big]$. These moments serve as measures that characterize the shape of the distribution of causal effects. Furthermore, we examine the central moments 
$\mathbb{E}\Big[\Big\{(Y_1-Y_0)-(\mathbb{E}[Y_1]-\mathbb{E}[Y_0])\Big\}^m\Big]$.
These moments quantify deviations from the ACE.
They encompass key statistical measures such as variance, standard deviation, skewness, and kurtosis, which are fundamental for characterizing the shape of the distribution of causal effects.
While previous work has examined  the second central moment (variance) of causal effects \citep{Heckman1997,Hernan2024}, this work provides a general analysis by studying arbitrary moments of causal effects.


Several studies \citep{DiNardo1996,Robins2001,Rubin2006,Jung2021,Kennedy2023d,Kim2024} aim to estimate the probability density function (PDF) of $Y_x$. 
However, identifying the moments of causal effects requires the joint distribution of $(Y_0, Y_1)$. 
The joint distribution of potential outcomes has been explored in the framework of probabilities of causation (PoC) \citep{Pearl1999, Tian2000, ALi2024}. 
We identify the (central) moments of causal effects  by leveraging the recent identification results for variants of the PoC established by  \citet{Kawakami2024}. 
Additionally, we derive bounds for the (central) moments of causal effects under relaxed assumptions. 

We further address the following causal question:
\begin{center}
(\textbf{Question 2}).
``{\it
How are two causal effects related?
}"
\end{center}
Researchers often consider more than two treatment options and compare multiple potential outcomes, $\{Y_1,Y_2,\dots,Y_R\}$,  
as discussed in 
\citep{Bartholomew1959, Page1963, Imbens2000, Imai2004}.
We then investigate the \emph{product moments} of causal effects 
$\mathbb{E}\Big[(Y_i-Y_j)(Y_k-Y_h)\Big]$,
as well as the central product moments (covariance and correlation) of causal effects,   
where $Y_i-Y_j$ represents the causal effect of changing $X=j$ to $X=i$, and $Y_k-Y_h$ represents the causal effect of changing $X=h$ to $X=k$. 
The product moments of causal effects reveal the association between two causal effects.
When $\mathbb{E}[(Y_i-Y_j)(Y_k-Y_h)]$ is positive, subjects with larger $Y_i-Y_j$ tend to have larger $Y_k-Y_h$.
When it is negative, subjects with larger $Y_i-Y_j$ tend to have smaller $Y_k-Y_h$.
When it is zero, there is no linear relationship between $Y_k-Y_h$ and $Y_i-Y_j$. 
The product moments of causal effects provide additional insights beyond the ACE $\mathbb{E}[Y_i-Y_j]$ and $\mathbb{E}[Y_k-Y_h]$.

We establish identification theorems for the (central) product moments of causal effects and derive bounds for them under more relaxed assumptions. 
Finally, we conduct experiments estimating the (product) moments of causal effects from finite samples and demonstrate their practical application using a real-world medical dataset. 

\section{Notations and Backgrounds}

We represent each variable with a capital letter $(X)$ and its realized value with a lowercase letter $(x)$. Let $\mathbb{I}(x)$ denote the indicator function, which takes the value $1$ if $x$ is true and $0$ otherwise.
We denote $\Omega_Y$ as the domain of $Y$, $\mathbb{E}[Y]$ as the expectation of $Y$, and $\mathbb{P}(Y < y)$ as the cumulative distribution function (CDF). 
Additionally, we denote the $m$-th Cartesian product of the domain $\Omega$ as $\Omega^m$, i.e., $\Omega^m = \Omega \times \Omega \times \dots \times \Omega$ (repeated $m$ times).

{\bf Moments of random variables.}
Moments are measures related to the shape of the probability distribution of a random variable. 
For each $m \geq 1$,
the $m$-th moment of a random variable $Y$ is defined by the expectation of the $m$-th power of $Y$, i.e., 
$\mathbb{E}[Y^m]$, and the central moment of a random variable is defined by
$C_m\defeq\mathbb{E}[(Y-\mathbb{E}[Y])^m]$.
Common statistics involving moments include mean $\mathbb{E}[Y]$, 
variance $C_2$, standard deviation $\displaystyle\sqrt{C_2}$, skewness $\displaystyle C_3/C_2^{3/2}$, and kurtosis $C_4/C_2^2$. 
They are essential statistics to capture the shape of the probability distribution of $Y$.
The product moments of the random variables $X$ and $Y$ are given by $\mathbb{E}[XY]$, the covariance is $\mathbb{E}[(X-\mathbb{E}[X])(Y-\mathbb{E}[Y])]$, and are used to define Pearson correlation coefficient  $\frac{\mathbb{E}[(X-\mathbb{E}[X])(Y-\mathbb{E}[Y])]}{\sqrt{\mathbb{E}[(X-\mathbb{E}[X])^2]}\sqrt{\mathbb{E}[(Y-\mathbb{E}[Y])^2]}}$.  

{\bf Structural causal models.} 
We use the language of Structural Causal Models (SCM) as our basic semantic and inferential framework \citep{Pearl09}.
An SCM ${\cal M}$ is a tuple $\left<{\boldsymbol U},{\boldsymbol V}, {\cal F}, \mathbb{P}_{\boldsymbol U} \right>$, where ${\boldsymbol U}$ is a set of exogenous (unobserved) variables following a joint distribution $\mathbb{P}_{\boldsymbol U}$, and ${\boldsymbol V}$ is a set of endogenous (observable) variables whose values are determined by structural functions ${\cal F}=\{f_{V_i}\}_{V_i \in {\boldsymbol V}}$ such that $v_i:= f_{V_i}({\mathbf{pa}}_{V_i},{\boldsymbol u}_{V_i})$ where ${\mathbf{PA}}_{V_i} \subseteq {\boldsymbol V}$ and $U_{V_i} \subseteq {\boldsymbol U}$. 
Each SCM ${\cal M}$ induces an observational distribution $\mathbb{P}_{\boldsymbol V}$ over ${\boldsymbol V}$. 
An atomic intervention of setting a set of endogenous variables ${\boldsymbol X}$ to constants ${\boldsymbol x}$, denoted by $do({\boldsymbol x})$, replaces the original equations of ${\boldsymbol X}$ by ${\boldsymbol X} :={\boldsymbol x}$ and induces a \textit{sub-model}  ${\cal M}_{\boldsymbol x}$.
We denote the potential outcome $Y$ under intervention $do({x})$ by $Y_{{x}}({\boldsymbol u})$, which is the solution of $Y$  in the sub-model ${\cal M}_{x}$ given ${\boldsymbol U}={\boldsymbol u}$.

{\bf Causal effects.}
Researchers usually consider the following SCM, denoted as ${\cal M}$:
\begin{equation}
\begin{gathered}
Y:=f_Y(X,U_Y),\ \  X:=f_X(U_X),
\end{gathered}
\end{equation}
where $U_Y$ and $U_X$ are latent exogenous variables.
The individual causal effect (ICE) is defined as $\text{\normalfont ICE}({\boldsymbol u})\defeq Y_1({\boldsymbol u})-Y_0({\boldsymbol u})$.
The average causal effect (ACE) is defined as $\text{\normalfont ACE}\defeq\mathbb{E}[Y_1-Y_0]$. 
\citet{Heckman1997} showed the identification of ICE under the rank invariance assumption stating that 
{``for almost every subject whose potential outcomes are $(y_1,y_0)=(Y_1,Y_0)$,  $\mathbb{P}(Y_0<y_0)=\mathbb{P}(Y_1<y_1)$ holds",} 
which is a strong assumption. 
They identified the variance of causal effects by identifying ICE. 
\citep{Hoshino2020} studied the identification  of the joint PDF of $(Y_0,Y_1)$ under various parametric specifications, whereas this work considers a nonparametric setting.

{\bf Joint distribution of potential outcomes.}
Joint distributions of potential outcomes, fundamental to this work, have been employed in the framework of probabilities of causation (PoC) \citep{Pearl1999,Tian2000,Li2024}. 
PoC are a family of probabilities quantifying whether one event was the real cause of another. 
Recently, \citet{Kawakami2024} defined the probability of necessity and sufficiency (PNS) for continuous treatment and outcome as $\mathbb{P}(Y_{x_0}< y \leq Y_{x_1})$, and {showed that it is identified from $\mathbb{P}(X,Y)$ if there are no unmeasured confounders and the function $f_Y(x,U_Y)$ satisfies a monotonicity assumption:}  
\begin{assumption}[Exogeneity]
\label{ASEXO2}
$Y_x\indep X$ for all $x \in \Omega_X$.
\end{assumption}
\begin{assumption}[Monotonicity over $f_Y$]
\label{MONO2}
{The function $f_Y(x,U_Y)$ is either (i) monotonic increasing on $U_Y$
for all $x \in \Omega_X$ almost surely w.r.t. $\mathbb{P}_{U_Y}$, or (ii) monotonic decreasing on $U_Y$
for all $x \in \Omega_X$
almost surely w.r.t. $\mathbb{P}_{U_Y}$.} 
\end{assumption}
{We will use the above assumptions to  identify moments of causal effects in this paper.}

\section{Moments of Causal Effects}

In this section, we study the (central) moments of causal effects $Y_1 - Y_0$ to address (\textbf{Question 1}). 

\subsection{Definition of the moments of causal effects} 
\label{sec-mce}
We define the moments of causal effects in the same manner as the moments of random variables.
\begin{definition}[The moments of causal effects]
For each $m\geq 1$,
the $m$-th moment of causal effect $Y_1-Y_0$ is defined as 
\begin{equation}
\mu^{(m)}\defeq\mathbb{E}\Big[(Y_1-Y_0)^m\Big].
\end{equation}
\end{definition}
The $m$-th moment of the causal effect is defined as the expectation of the $m$-th power of the causal effect $Y_1 - Y_0$.
The first moment, $\mu^{(1)}=\mathbb{E}[Y_1 - Y_0]$, is the ACE.

We present two examples to illustrate what the moments of causal effects specifically measure in simple SCMs.

{\bf Example 1.} (Homogeneous ICE)
Consider a simple linear SCM given by $Y=X+U_Y$ where $\mathbb{E}[U_Y]=0$.
In this model, the ICE is equal to $1$ for every subject, meaning that the causal effect is homogeneous.
Consequently, all of the $m$-th moments of the causal effect are also equal to $1$.
Note that 
$\mathbb{E}[Y_1^m-Y_0^m]=\mathbb{E}[(1+U_Y)^m-U_Y^m]$ {varies with $m$.} 

{\bf Example 2.} (Heterogeneous ICE)
Consider a linear SCM with an interaction term between $X$ and $U_Y$, $Y=X(U_Y+1)+1$ where $\mathbb{E}[U_Y]=0$.
In this model, $\text{ICE} = Y_1 - Y_0 = U_Y+1$, which varies across subjects, making the causal effect heterogeneous.
The $m$-th moment of the causal effect  is given by $\mathbb{E}[(U_Y+1)^m]$.
In comparison, 
$\mathbb{E}[Y_1^m-Y_0^m]=\mathbb{E}[(U_Y+2)^m-1]$.

The central moments of causal effects are defined as the moments of causal effects measured relative to their mean.
\begin{definition}[The central moment of causal effects]
For each $m\geq 1$,
the $m$-th central moment of causal effect $Y_1-Y_0$ is defined as
\begin{equation}
\overline{\mu}^{(m)}\defeq\mathbb{E}\Big[\Big\{(Y_1-Y_0)-(\mathbb{E}[Y_1]-\mathbb{E}[Y_0])\Big\}^m\Big].
\end{equation}
\end{definition}
The first central moment of causal effects is always $0$ since we have $\overline{\mu}^{(1)}=\mathbb{E}[(Y_1-Y_0)]-\mathbb{E}[(\mathbb{E}[Y_1]-\mathbb{E}[Y_0])]=0$.
When $m\geq 2$, $\overline{\mu}^{(m)}$ is not equal to $\mathbb{E}[(Y_1-\mathbb{E}[Y_1])^m]-\mathbb{E}[(Y_0-\mathbb{E}[Y_0])^m]$ as discussed in \citep{Wiedermann2022}.

We revisit Examples 1 and 2 in Section~\ref{sec-mce}  to illustrate  the central moments of causal effects. 

{\bf Example 1 (continued).} (Homogeneous ICE)
In SCM given by $Y=X+U_Y$ where $\mathbb{E}[U_Y]=0$, the $m$-th central moment is equal to $0$ for any $m \geq 1$, indicating that the causal effect is homogeneous.

{\bf Example 2 (continued).} (Heterogeneous ICE)
In SCM $Y=X(U_Y+1)+1$ where $\mathbb{E}[U_Y]=0$, the $m$-th central moment is given by $\mathbb{E}[U_Y^m]$ for all $m \geq 1$.
The central moments of causal effects correspond to the (central) moments of the random variable $U_Y$.

The higher order of moments of causal effects may  provide useful information on the distribution. The central moments of causal effects can be used to compute important and well-known statistics such as the variance $\overline{\mu}^{(2)}$, standard deviation $\sqrt{\overline{\mu}^{(2)}}$, skewness ${\overline{\mu}^{(3)}}/{{\overline{\mu}^{(2)}}^{3/2}}$, and kurtosis ${\overline{\mu}^{(4)}}/{{\overline{\mu}^{(2)}}^{2}}$ of the causal effects.
Variance and standard deviation quantify the dispersion of a distribution.
If the variance of causal effects is large, the causal effects may deviate significantly from ACE for some subjects. 
When the variance of causal effects is small, ICE is close to ACE for all subjects. 
Skewness is a measure of the asymmetry of a probability distribution. 
If the causal effect is positively skewed, the right tail of the distribution of the causal effect is longer.
If the causal effect is negatively skewed, the left tail of the distribution is longer.
Kurtosis is a measure of the tailedness or peakedness of a distribution.  
High kurtosis values indicate the presence of outliers in causal effects \citep{Westfall2014}.

\subsection{Identification of the moments of causal effects}

Under 
the exogeneity assumption (Assumption \ref{ASEXO2}), the first moment of causal effects is identifiable as $\mathbb{E}[Y|X=1]-\mathbb{E}[Y|X=0]$ \citep{Holland1986}.
In this section, we discuss the identification  of the higher moments of causal effects. 

When $m \geq 2$, $\mu^{(m)}$ is not equal to $\mathbb{E}[Y_1^m - Y_0^m]$, as discussed in \citep{Hernan2024, Kuroki2024}. Under Assumption \ref{ASEXO2}, $\mathbb{E}[Y_1^m - Y_0^m]$ is identifiable as $\mathbb{E}[Y^m | X = 1] - \mathbb{E}[Y^m | X = 0]$. However, $\mu^{(m)}$ remains unidentifiable.  
For example, the second moment of the causal effect, $\mu^{(2)}$, is given by $\mathbb{E}[Y_1^2]-2\mathbb{E}[Y_0 Y_1]+\mathbb{E}[Y_0^2]$, where the term $\mathbb{E}[Y_0 Y_1]$ is not identifiable. 

To prepare  the identification of the moments of causal effect,
 we first decompose $(Y_1-Y_0)^m$ into two parts as follows:
\begin{lemma}
\label{lem1}
Under SCM ${\cal M}$, 
we have
\begin{align}
\label{eq5}
&(Y_1-Y_0)^m=(Y_1-Y_0)^m\mathbb{I}(Y_1>Y_0)\nonumber\\
&\hspace{3cm}+(-1)^m(Y_0-Y_1)^m\mathbb{I}(Y_0>Y_1)\nonumber\\
&=\int_{{\Omega_Y}^m} \mathbb{I}(Y_0<y_1\leq Y_1,Y_0<y_2\leq Y_1,\dots,\nonumber\\
&\hspace{3.5cm}Y_0<y_m\leq Y_1)dy_1\dots dy_m\nonumber\\
&+(-1)^m\int_{{\Omega_Y}^m} \mathbb{I}(Y_1<y_1\leq Y_0,Y_1<y_2\leq Y_0,\dots,\nonumber\\
&\hspace{3.5cm}Y_1<y_m\leq Y_0)dy_1\dots dy_m.
\end{align}
\end{lemma}

The first part corresponds to subjects with a positive ICE, where $Y_1 - Y_0 > 0$, and the second part corresponds to subjects with a negative ICE, where $Y_1 - Y_0 < 0$.




We make the following assumption:
\begin{assumption}[Finiteness of integrals]
\label{exi1}
Under SCM ${\cal M}$, for $m\geq 1$, 
$\mu^{(m)}<\infty$ and
$\int_{{\Omega_Y}^m} \mathbb{P}(Y_i<y_1\leq Y_j,Y_i<y_2\leq Y_j,\dots,Y_i<y_m\leq Y_j)dy_1\dots dy_m<\infty$ hold for $(i,j)=\{(0,1),(1,0)\}$.
\end{assumption}

Under SCM ${\cal M}$ and Assumption \ref{exi1},  taking the expectation on both sides of  Eq.~\eqref{eq5}, we have
\begin{align}
\label{eq6}
\mu^{(m)}
&=\int_{{\Omega_Y}^m} \mathbb{P}(Y_0<y_1\leq Y_1,Y_0<y_2\leq Y_1,\dots,\nonumber\\
&\hspace{3cm}Y_0<y_m\leq Y_1)dy_1\dots dy_m\nonumber\\
&+(-1)^m\int_{{\Omega_Y}^m} \mathbb{P}(Y_1<y_1\leq Y_0,Y_1<y_2\leq Y_0,\dots,\nonumber\\
&\hspace{3cm}Y_1<y_m\leq Y_0)dy_1\dots dy_m.
\end{align}
The identification of the moments of causal effects then reduces to the identification  of  $\mathbb{P}(Y_0<y_1\leq Y_1,Y_0<y_2\leq Y_1,\dots,Y_0<y_m\leq Y_1)$ and $\mathbb{P}(Y_1<y_1\leq Y_0,Y_1<y_2\leq Y_0,\dots,Y_1<y_m\leq Y_0)$. The identification  of this type of  joint distributions of potential outcomes was discussed in \citep{Kawakami2024}, based on which we obtain the following result: 
\begin{theorem}[Identification of the moments of causal effect]
\label{theo1}
Under SCM ${\cal M}$ and Assumptions \ref{ASEXO2}, \ref{MONO2}, and \ref{exi1},  the $m$-th moment of causal effect $Y_1-Y_0$ is identifiable by $\mu^{(m)}=\sigma^{(m)}$, where
\begin{align}
\label{eq10}
&\sigma^{(m)}=\int_{{\Omega_Y}^m} \max\Big\{\min_{p=1,\dots,m}\{\mathbb{P}(Y<y_p|X=0)\}\nonumber\\
&\hspace{0.2cm}-\max_{p=1,\dots,m}\{\mathbb{P}(Y<y_p|X=1)\},0\Big\}dy_1\dots dy_m\nonumber\\
&+(-1)^m\int_{{\Omega_Y}^m} \max\Big\{\min_{p=1,\dots,m}\{\mathbb{P}(Y<y_p|X=1)\}\nonumber\\
&\hspace{0.2cm}-\max_{p=1,\dots,m}\{\mathbb{P}(Y<y_p|X=0)\},0\Big\}dy_1\dots dy_m.
\end{align}
\end{theorem}
Theorem~\ref{theo1} says that the moments of causal effects can be expressed in terms of  conditional CDFs. 
For $m=1$, Eq.~\eqref{eq10} reduces to ACE$=\mathbb{E}[Y_1-Y_0]=\int_{\Omega_Y}\{\mathbb{P}(Y<y_1|X=0)-\mathbb{P}(Y<y_1|X=1)\}dy_1$   
\citep{Ju2010}, 
which  does not require Assumption \ref{MONO2} to hold.

For $m=2$, the second moment of causal effects (variance) $\mu^{(2)}$ is given by $\int_{\Omega_Y} \int_{\Omega_Y} \max\{\min\{\mathbb{P}(Y<y_1|X=0),\mathbb{P}(Y<y_2|X=0)\}-\max\{\mathbb{P}(Y<y_1|X=1),\mathbb{P}(Y<y_1|X=1)\},0\}dy_1dy_2
+\int_{\Omega_Y} \int_{\Omega_Y} 
 \max\{\min\{\mathbb{P}(Y<y_1|X=1),\mathbb{P}(Y<y_2|X=1)\}-\max\{\mathbb{P}(Y<y_1|X=0),\mathbb{P}(Y<y_2|X=0),0\}dy_1dy_2$.

\subsection{Bounding the moments of causal effects}

The monotonicity Assumption \ref{MONO2} may sometimes be considered implausible by researchers.
Therefore, we derive bounds for the moments of causal effects that do not  rely on Assumption \ref{MONO2}.

We first provide bounds of the joint distribution of the potential outcomes $\mathbb{P}(Y_i<y_1\leq Y_j,Y_i<y_2\leq Y_j,\dots,Y_i<y_m\leq Y_j)$ using Fr\'{e}chet inequalities \citep{Frechet1935,Frechet1960}.
\begin{lemma}
\label{lem3}
Under SCM ${\cal M}$ and Assumptions \ref{ASEXO2} and \ref{exi1}, 
we have $l(y_1,\dots,y_m;i,j)\leq \mathbb{P}(Y_j<y_1\leq Y_i,Y_j<y_2\leq Y_i,\dots,Y_j<y_m\leq Y_i) \leq u(y_1,\dots,y_m;i,j)$, 
where 
\begin{align}
&l(y_1,\dots,y_m;i,j)=\max\Big\{\sum_{p=1,\dots,m}\mathbb{P}(Y<y_p|X=j)\nonumber\\
&\hspace{0.5cm}-\sum_{p=1,\dots,m}\mathbb{P}(Y<y_p|X=i)-m+1,0\Big\},\\
&u(y_1,\dots,y_m;i,j)=\min\Big\{\min_{p=1,\dots,m}\{\mathbb{P}(Y<y_p|X=j)\},\nonumber\\
&\hspace{1.5cm}1-\max_{p=1,\dots,m}\{\mathbb{P}(Y<y_p|X=i)\}\Big\}
\end{align}
for $(i,j) \in \{(1,0),(0,1)\}$ and any $y_1, \dots, y_m \in \Omega_Y$.
\end{lemma}

Then, we have the following theorem.
\begin{theorem}[Bounds of the moments of causal effect]
\label{theo2}
Under SCM ${\cal M}$ and Assumptions \ref{ASEXO2} and \ref{exi1},  we have $\sigma_L^{(m)} \leq \mu^{(m)} \leq \sigma_U^{(m)}$, where

(A). When $m$ is an even number,
\begin{align}
\label{eq13}
&\sigma_L^{(m)}=\int_{{\Omega_Y}^m} l(y_1,\dots,y_m;1,0)dy_1\dots dy_m\nonumber\\
&\hspace{1.5cm}+\int_{{\Omega_Y}^m} l(y_1,\dots,y_m;0,1)dy_1\dots dy_m,\\
&\sigma_U^{(m)}=\int_{{\Omega_Y}^m} u(y_1,\dots,y_m;1,0)dy_1\dots dy_m\nonumber\\
&\hspace{1.5cm}+\int_{{\Omega_Y}^m} u(y_1,\dots,y_m;0,1)dy_1\dots dy_m.
\end{align}

(B). When $m$ is an odd number,
\begin{align}
&\sigma_L^{(m)}=\int_{{\Omega_Y}^m} l(y_1,\dots,y_m;1,0)dy_1\dots dy_m\nonumber\\
&\hspace{1.5cm}-\int_{{\Omega_Y}^m} u(y_1,\dots,y_m;0,1)dy_1\dots dy_m,\\
\label{eq16}
&\sigma_U^{(m)}=\int_{{\Omega_Y}^m} u(y_1,\dots,y_m;1,0)dy_1\dots dy_m\nonumber\\
&\hspace{1.5cm}-\int_{{\Omega_Y}^m} l(y_1,\dots,y_m;0,1)dy_1\dots dy_m.
\end{align}
\end{theorem}
If $\sigma_U^{(m)}=\infty$ and $\sigma_L^{(m)}=-\infty$, then $\mu^{(m)}$ is unbounded.

{The upper bound of the Fr\'{e}chet inequalities is always sharp for all $m\geq 1$ \citep{Nelsen2007}; thus, the function $u(y_1,\dots,y_m;i,j)$ in Lemma 2 is sharp for all $m\geq 1$.
In contrast, the lower bound of the Fr\'{e}chet inequalities is not always sharp  except when $m=1$; hence, the function $l(y_1,\dots,y_m;i,j)$ in Lemma 2 is not sharp. 
As a result, only the upper bounds of the moments of causal effects are sharp when $m$ is even.
In all other cases, our bounds of the moments of causal effects are not sharp.}

\textbf{Remark.} We present a similar identification theorem and bounds on the central moments of causal effects in Appendix \ref{appB}.
Additionally, the skewness and kurtosis of causal effects are also bounded, as shown in Appendix \ref{appB}.

\section{Product Moments of Causal Effects}

In this section, we study the (central) product moments of causal effects to address (\textbf{Question 2}).

Let $\Omega_X=\{1,\dots,R\}$. 
The causal effect of changing $X=j$ to $X=i$ is given by $Y_i-Y_j$ and the causal effect of changing $X=k$ to $X=h$ is given by $Y_h-Y_k$.
We study the 
association of the two causal effects $Y_i-Y_j$ and $Y_h-Y_k$.

\subsection{Definition of the product moment of causal effects}

We define the product moment of two causal effects analogously to the product moment of two random variables.
\begin{definition}[The product moment of causal effects]
The product moment of causal effects is defined by 
\begin{equation}
\rho_{i,j;k,h}\defeq\mathbb{E}\Big[(Y_i-Y_j)(Y_k-Y_h)\Big].
\end{equation}
\end{definition}

We present three examples to illustrate the product moments of causal effects 
in simple SCMs.

{\bf Example 1 (continued).} (Homogeneous ICE)
In SCM given by $Y=X+U_Y$ where $\mathbb{E}[U_Y]=0$, 
the product moment of $Y_1-Y_0$ and $Y_0-Y_{-1}$ is equal to $1$.

{\bf Example 2 (continued).} (Heterogeneous ICE)
In SCM $Y=X(U_Y+1)+1$ where $\mathbb{E}[U_Y]=0$, 
the product moment of $Y_1-Y_0$ and $Y_0-Y_{-1}$ is equal to $\mathbb{E}[(U_Y+1)^2]>0$.

{\bf Example 3.} (Heterogeneous and nonlinear ICE)
We consider a nonlinear SCM with an interaction term between $X^2$ and $U_Y$:  $Y=X^2(U_Y+1)+1$ where $\mathbb{E}[U_Y]=0$. We have  $Y_1-Y_0=U_Y+1$ and $Y_0-Y_{-1}=-(U_Y+1)$ and are heterogeneous.
The product moment of $Y_1-Y_0$ and $Y_0-Y_{-1}$ is equal to $\mathbb{E}[-(U_Y+1)^2]<0$.




We examine the covariance and correlation of two causal effects. 
\begin{definition}[Covariance of causal effects]
We define the covariance (central product moment of causal effects) as 
\begin{equation}
\begin{aligned}
\overline{\rho}_{i,j;k,h}\defeq&\mathbb{E}\Big[\Big\{(Y_i-Y_j)-(\mathbb{E}[Y_i]-\mathbb{E}[Y_j])\Big\}\\
&\hspace{0.5cm}\times\Big\{(Y_k-Y_h)-(\mathbb{E}[Y_k]-\mathbb{E}[Y_h])\Big\}\Big].
\end{aligned}
\end{equation}
\end{definition}
\begin{definition}[Correlation of causal effects]
We define the correlation of causal effects as 
\begin{equation}
\begin{aligned}
&\overline{\tau}_{i,j;k,h}=\overline{\rho}_{i,j;k,h} \\
&\hspace{0cm}\Bigg/\bigg\{\sqrt{\mathbb{E}\Big[\Big\{(Y_i-Y_j)-(\mathbb{E}[Y_i]-\mathbb{E}[Y_j])\Big\}^2\Big]}\\
&\hspace{0.5cm}\times\sqrt{\mathbb{E}\Big[\Big\{(Y_k-Y_h)-(\mathbb{E}[Y_k]-\mathbb{E}[Y_h])\Big\}^2\Big]}
\bigg\}.
\end{aligned}
\end{equation}
\end{definition}
Correlation is a measure of association between two variables \citep{Pearson1905}. Similarly, the correlation of causal effects quantifies the association between two causal effects.


For instance, 
when comparing three treatments $X = 0,1,2$, researchers often evaluate their respective averages $\mathbb{E}[Y_0], \mathbb{E}[Y_1],$ and $\mathbb{E}[Y_2]$.
When $\mathbb{E}[Y_1 - Y_0] > 0$ and $\mathbb{E}[Y_2 - Y_1] > 0$, it is concluded that, on average, both changes, from $X = 0$ to $X = 1$ and from $X = 1$ to $X = 2$, have positive effects.
The correlation of causal effects provides more detailed insights.
When the correlation between $Y_1 - Y_0$ and $Y_2 - Y_1$ is negative, patients with larger causal effects $Y_1 - Y_0$ than the average tend to have smaller causal effects $Y_2 - Y_1$ than the average.
Conversely, when patients have smaller causal effects $Y_1 - Y_0$ than the average, they tend to have larger causal effects $Y_2 - Y_1$ than the average.

We present three examples to illustrate  the covariance and correlation of causal effects. 

{\bf Example 1 (continued).} (Homogeneous ICE)
In SCM given by $Y=X+U_Y$ where $\mathbb{E}[U_Y]=0$, 
the covariance of $Y_1-Y_0$ and $Y_0-Y_{-1}$ is equal to $0$.

{\bf Example 2 (continued).} (Heterogeneous ICE)
In SCM $Y=X(U_Y+1)+1$ where $\mathbb{E}[U_Y]=0$, 
the covariance of $Y_1-Y_0$ and $Y_0-Y_{-1}$ is equal to $\mathbb{E}[U_Y^2]>0$, 
 the correlation of $Y_1-Y_0$ and $Y_0-Y_{-1}$ is $1$, and they have a positive correlation.

{\bf Example 3 (continued).} (Heterogeneous and nonlinear ICE)
In SCM $Y=X^2(U_Y+1)+1$ where $\mathbb{E}[U_Y]=0$, the covariance of $Y_1-Y_0$ and $Y_0-Y_{-1}$ is equal to $\mathbb{E}[-U_Y^2]<0$, 
 the correlation of $Y_1-Y_0$ and $Y_0-Y_{-1}$ is $-1$,  and they have a negative correlation.

\subsection{Identification of the product moment of causal effects}

To prepare the discussion on the identification  of the product moment of causal effects, we decompose the product of causal effects $(Y_i-Y_j)$ and $(Y_k-Y_h)$ into four parts.
\begin{lemma}
\label{lem4}
Under SCM ${\cal M}$, 
we have
\begin{align}
\label{eq20}
&(Y_i-Y_j)(Y_k-Y_h)\nonumber\\
&=(Y_i-Y_j)(Y_k-Y_h)\mathbb{I}(Y_i>Y_j,Y_k>Y_h)\nonumber\\
&-(Y_j-Y_i)(Y_k-Y_h)\mathbb{I}(Y_j>Y_i,Y_k>Y_h)\nonumber\\
&-(Y_i-Y_j)(Y_h-Y_k)\mathbb{I}(Y_i>Y_j,Y_h>Y_k)\nonumber\\
&+(Y_j-Y_i)(Y_h-Y_k)\mathbb{I}(Y_j>Y_i,Y_h>Y_k)\nonumber\\
&=\int_{\Omega_Y}\int_{\Omega_Y} \mathbb{I}(Y_j<y_1\leq Y_i,Y_h<y_2\leq Y_k)dy_1 dy_2\nonumber\\
&-\int_{\Omega_Y}\int_{\Omega_Y} \mathbb{I}(Y_i<y_1\leq Y_j,Y_h<y_2\leq Y_k)dy_1 dy_2\nonumber\\
&-\int_{\Omega_Y}\int_{\Omega_Y} \mathbb{I}(Y_j<y_1\leq Y_i,Y_k<y_2\leq Y_h)dy_1 dy_2\nonumber\\
&+\int_{\Omega_Y}\int_{\Omega_Y} \mathbb{I}(Y_i<y_1\leq Y_j,Y_k<y_2\leq Y_h)dy_1 dy_2.
\end{align}
\end{lemma}
The above decomposition consists of four parts based on the signs of ICE. 
We make the following assumption:
\begin{assumption}[Finiteness of integrals]
\label{exi2}
Under SCM ${\cal M}$,
$\rho_{i,j;k,h}<\infty$ and 
$\int_{\Omega_Y}\int_{\Omega_Y} \mathbb{P}(Y_j<y_1\leq Y_i,Y_h<y_2\leq Y_k)dy_1 dy_2<\infty$ hold for any $i,j,k,h \in \{1,\dots,R\}$.
\end{assumption}

Under SCM ${\cal M}$ and Assumption \ref{exi2}, 
taking the expectation on both sides of  Eq.~\eqref{eq20}, we have
\begin{align}
\label{eq21}
&\rho_{i,j;k,h}\nonumber\\
&=\int_{\Omega_Y}\int_{\Omega_Y} \mathbb{P}(Y_j<y_1\leq Y_i,Y_h<y_2\leq Y_k)dy_1 dy_2\nonumber\\
&-\int_{\Omega_Y}\int_{\Omega_Y} \mathbb{P}(Y_i<y_1\leq Y_j,Y_h<y_2\leq Y_k)dy_1 dy_2\nonumber\\
&-\int_{\Omega_Y}\int_{\Omega_Y} \mathbb{P}(Y_j<y_1\leq Y_i,Y_k<y_2\leq Y_h)dy_1 dy_2\nonumber\\
&+\int_{\Omega_Y}\int_{\Omega_Y} \mathbb{P}(Y_i<y_1\leq Y_j,Y_k<y_2\leq Y_h)dy_1 dy_2.
\end{align}

The identification of joint distributions of  potential outcomes in the form of  $\mathbb{P}(Y_j<y_1\leq Y_i,Y_h<y_2\leq Y_k)$ was discussed in \citep{Kawakami2024}, based on which we obtain the following result:
\begin{theorem}[Identification of the product moments of causal effects]
\label{theo3}
Under SCM ${\cal M}$ and Assumptions \ref{ASEXO2}, \ref{MONO2}, and \ref{exi2}, 
the product moment of $(Y_i-Y_j)$ and $(Y_k-Y_h)$ is identifiable by $\rho_{i,j;k,h}=\sigma(i,j;k,h)$, where
\begin{align}
\label{eq22}
&\sigma(i,j;k,h)\nonumber\\
&=\int_{{\Omega_Y}^2}\max\Big\{\min\{\mathbb{P}(Y<y_1|X=j),\mathbb{P}(Y<y_2|X=h)\}\nonumber\\
&\hspace{0.2cm}-\max\{\mathbb{P}(Y<y_1|X=i),\mathbb{P}(Y<y_2|X=k)\},0\Big\}dy_1dy_2\nonumber\\
&-\int_{{\Omega_Y}^2}\max\Big\{\min\{\mathbb{P}(Y<y_1|X=i),\mathbb{P}(Y<y_2|X=h)\}\nonumber\\
&\hspace{0.2cm}-\max\{\mathbb{P}(Y<y_1|X=j),\mathbb{P}(Y<y_2|X=k)\},0\Big\}dy_1dy_2\nonumber\\
&-\int_{{\Omega_Y}^2}\max\Big\{\min\{\mathbb{P}(Y<y_1|X=j),\mathbb{P}(Y<y_2|X=k)\}\nonumber\\
&\hspace{0.2cm}-\max\{\mathbb{P}(Y<y_1|X=i),\mathbb{P}(Y<y_2|X=h)\},0\Big\}dy_1dy_2\nonumber\\
&+\int_{{\Omega_Y}^2}\max\Big\{\min\{\mathbb{P}(Y<y_1|X=i),\mathbb{P}(Y<y_2|X=k)\}\nonumber\\
&\hspace{0.2cm}-\max\{\mathbb{P}(Y<y_1|X=j),\mathbb{P}(Y<y_2|X=h)\},0\Big\}dy_1dy_2.
\end{align}
\end{theorem}

\subsection{Bounding the product moment of causal effects}

The monotonicity Assumption \ref{MONO2} may sometimes be considered implausible by researchers. 
Therefore, we derive bounds for the product moment of causal effects that do not rely on Assumption \ref{MONO2}.


We first derive  bounds for the joint distribution of  potential outcomes $\mathbb{P}(Y_j<y_1\leq Y_i,Y_h<y_2\leq Y_k)$.
\begin{lemma}
\label{lem6}
Under SCM ${\cal M}$ and Assumptions \ref{ASEXO2} and \ref{exi2}, for any $i,j,k,h \in \{1,\dots,R\}$, $y_1, y_2 \in \Omega_Y$
we have $l(y_1,y_2;i,j,k,h)\leq \mathbb{P}(Y_j<y_1\leq Y_i,Y_h<y_2\leq Y_k)\leq u(y_1,y_2;i,j,k,h)$, where
\begin{align}
&l(y_1,y_2;i,j,k,h)=\nonumber\\
&\hspace{0cm}\max\Big\{\mathbb{P}(Y<y_1|X=j)-\mathbb{P}(Y <y_1|X=i)\nonumber\\
&\hspace{0cm}+\mathbb{P}(Y<y_2|X=h)-\mathbb{P}(Y<y_2|X=k)-1,0\Big\},\\
&u(y_1,y_2;i,j,k,h)=\nonumber\\
&\hspace{0cm}\min\Big\{\min\{\mathbb{P}(Y<y_1|X=j),\mathbb{P}(Y<y_2|X=h)\},\nonumber\\
&\hspace{0cm}1-\max\{\mathbb{P}(Y<y_1|X=i),\mathbb{P}(Y<y_2|X=k)\}\Big\}
\end{align}
\end{lemma}
Then, we have the following theorem.
\begin{theorem}[Bounds of the product moments of causal effects]
\label{theo4}
Under SCM ${\cal M}$ and Assumptions \ref{ASEXO2} and \ref{exi2}, for any $i,j,k,h \in \{1,\dots,R\}$, we have $\sigma_L(i,j;k,h) \leq \rho_{i,j;k,h} \leq \sigma_U(i,j;k,h)$, where
\begin{align}
\label{eq25}
\sigma_L(i,j;k,h)\nonumber
&=\int_{{\Omega_Y}^2}l(y_1,y_2;i,j,k,h)dy_1dy_2\nonumber\\
&-\int_{{\Omega_Y}^2}u(y_1,y_2;j,i,k,h)dy_1dy_2\nonumber\\
&-\int_{{\Omega_Y}^2}u(y_1,y_2;i,j,h,k)dy_1dy_2\nonumber\\
&+\int_{{\Omega_Y}^2}l(y_1,y_2;j,i,h,k)dy_1dy_2,\\
\label{eq26}
\sigma_U(i,j;k,h)\nonumber
&=\int_{{\Omega_Y}^2}u(y_1,y_2;i,j,k,h)dy_1dy_2\nonumber\\
&-\int_{{\Omega_Y}^2}l(y_1,y_2;j,i,k,h)dy_1dy_2\nonumber\\
&-\int_{{\Omega_Y}^2}l(y_1,y_2;i,j,h,k)dy_1dy_2\nonumber\\
&+\int_{{\Omega_Y}^2}u(y_1,y_2;j,i,h,k)dy_1dy_2.
\end{align}
\end{theorem}
If $\sigma_U=\infty$ and $\sigma_L=-\infty$, then $\rho_{i,j;k,h}$ is unbounded.
{The bounds for product moments are not sharp.}

\textbf{Remark.}  We present a similar identification theorem and bounds for the central product moment (covariance) of causal effects in Appendix \ref{appC}. 
Additionally, the correlation of causal effects is also bounded, as shown in Appendix \ref{appC}.

\section{Numerical Experiments}
\label{sec-exp}
In this section, we perform experiments to illustrate the finite-sample performance of the estimators of the moments of causal effects.


{\bf Estimation.}
The family of moments of causal effects $\sigma^{(m)}, \sigma_L^{(m)}, \sigma_U^{(m)}, \sigma(i,j;k,h), \sigma_L(i,j;k,h), \sigma_U(i,j;k,h)$ in Theorems~\ref{theo1}-\ref{theo4} (and the central moments of causal effects $\bar{\sigma}^{(m)},\bar{\sigma}_L^{(m)}, \bar{\sigma}_U^{(m)}, \bar{\sigma}(i,j,k,h), \bar{\sigma}_U(i,j,k,h), \bar{\sigma}_L(i,j,k,h)$ in Appendices \ref{appB} and \ref{appC}) are estimable by plugging in the empirical CDFs and expectations \citep{Vaart1998} and calculating the integrals using the Monte Carlo integration method \citep{Press2007}.
Let $N$ be the sample size of the dataset, and let $N_1$, $N_2$, $N_3$, and $N_4$ be the numbers of points for Monte Carlo integration on $y_1$, $y_2$, $y_3$, and $y_4$.
We assume that
the domains of $Y$ and $Y-\mathbb{E}[Y|X=x]$ for any $x \in \Omega_X$ are bounded by $[a,b]$, 
which is required for Monte Carlo integration. 
The details of all estimators are shown in Appendix \ref{appCon}.
They are all consistent estimators as discussed in Appendix \ref{appCon}.


{\bf Simulation for the moments of causal effects.}
We assume the following SCM (A):
\begin{gather}
Y:=-(X+1)U\mathbb{I}(XU\geq 0),\\
X \sim \text{Bern}({0.8}), U\sim \text{Unif}(-1,1),
\end{gather}
where $\text{Bern}(p)$ is a Bernoulli distribution with probability $p$, and $\text{Unif}(-1,1)$ is a uniform distribution over $[-1,1]$. 
This setting satisfies Assumptions \ref{ASEXO2} and \ref{MONO2}.
We simulate 1000 times with the sample size $N=20,100,1000$, respectively.
We let $N_1$, $N_2$, $N_3$, and $N_4$ all be 1000.


\begin{table*}[tb]
\centering
\caption{Results of numerical experiments for SCM (A).
{We present the estimates of the second moments $\sigma^{(2)}$, third moments $\sigma^{(3)}$, and fourth moments $\sigma^{(4)}$ of causal effects along with their respective upper and lower bounds.
Additionally, we report the means of each estimator accompanied by their 95\% confidence intervals.}}
\label{tab:a2}
\vspace{-0.25cm}
\scalebox{1}{
\begin{tabular}{c|cccc}
\hline
Estimators & $N=20$ & $N=100$ & $N=1000$ &  Ground Truth \\
\hline
\hline
$\sigma^{(2)}$  & $0.405 ([0.138,0.841])$ &  $0.373 ([0.215,0.659])$ & $0.335([0.289,0.418])$ &$0.333$ \\
$\sigma_U^{(2)}$  & $1.5478 ([0.804,2.623])$ &$1.582 ([1.127,2.030])$ &  $1.647 ([1.485,1.769])$ &- \\
$\sigma_L^{(2)}$   & $0.108 ([0.000,0.679])$ &$0.005 ([0.000,0.018])$ &  $0.000 ([0.000,0.000])$ &- \\
\hline
$\sigma^{(3)}$   &  $-0.572 ([-1.679,0.093])$ & $-0.293 ([-0.750,-0.065])$ & $-0.245 ([-0.305,-0.186])$ &$-0.250$ \\
$\sigma_U^{(3)}$   & $0.087 ([-0.107,0.292])$ &$0.120 ([0.037,0.234])$ &  $0.126 ([0.074,0.177])$ &- \\
$\sigma_L^{(3)}$   & $-2.479 ([-5.381,-0.612])$ &$-3.175 ([-3.999,-1.978])$ &  $-3.412 ([-3.877,-3.113])$ &- \\
\hline
$\sigma^{(4)}$   &  $0.963 ([0.066,6.922])$ & $0.194 ([0.061,0.384])$& $0.205 ([0.112,0.283])$ &  $0.200$ \\
$\sigma_U^{(4)}$  & $6.837 ([1.157,11.497])$ &$7.712 ([4.909,10.848])$ &  $8.093 ([7.214,8.878])$ &- \\
$\sigma_L^{(4)}$   & $0.008 ([0.000,0.085])$ &$0.000 ([0.000,0.000])$ &  $0.000 ([0.000,0.000])$ &- \\
\hline
\end{tabular}
}
\end{table*}

\begin{table*}[tb]
\centering
\caption{Results of numerical experiments for SCM (B).
{We present the estimates of the product moments of causal effects $\sigma(1,0;0,-1)$ along with their respective upper and lower bounds.
Additionally, we report the means of each estimator accompanied by their 95\% confidence intervals.}}
\label{tab:a}
\vspace{-0.25cm}
\scalebox{1}{
\begin{tabular}{c|cccc}
\hline
Estimators & $N=20$ & $N=100$ & $N=1000$ &  Ground Truth \\
\hline
\hline
$\sigma(1,0;0,-1)$  & $-0.300 ([-0.437,-0.131])$ & $-0.323 ([-0.420,-0.239])$ & $-0.327 ([-0.419,-0.260])$ &$-0.333$ \\
$\sigma_U(1,0;0,-1)$   & $-0.154 ([-0.521,0.000])$ &$-0.105 ([-0.217,-0.029])$ &  $-0.168 ([-0.222,-0.112])$ &- \\
$\sigma_L(1,0;0,-1)$   & $-0.352 ([-0.559,-0.100])$ &$-0.390 ([-0.583,-0.278])$  &  $-0.338 ([-0.409,-0.260])$ &- \\
\hline
\end{tabular}
}
\end{table*}

{\bf Results.}
We present the estimates obtained using our proposed methods
in Table \ref{tab:a2}.
All means of the estimators are close to the ground truth for $N=1000$. 
All ground truth values lie within the computed bounds. 
However, estimators for small sample sizes have large 95 $\%$ CIs, 
especially for high-order moments.

{\bf Simulation for the product moments of causal effects.}
We assume the following SCM (B):
\begin{equation}
Y:=X^2U, U\sim \text{Unif}(0,1),
\end{equation}
where $X$ takes values in $\{-1, 0, 1\}$ with the probabilities $\mathbb{P}(X=-1)=\mathbb{P}(X=0)=\mathbb{P}(X=1)=1/3$.  
The domain of $Y$ is bounded within $[0,1]$.
This setting satisfies Assumptions \ref{ASEXO2} and \ref{MONO2}.
We simulate 1000 times with the sample size $N=20,100,1000$, respectively.
We let $N_1$ and $N_2$ both be 1000.

{\bf Results.}
We present the estimates for $\mathbb{E}[(Y_1-Y_0)(Y_0-Y_{-1})]$ 
in Table \ref{tab:a}.
All means of the estimators are close to the ground truth. 
The ground truth value lies within the computed bounds.
However, estimators for small sample sizes have large 95 $\%$ CIs.

Overall, the results show that, as the sample size increases, the estimates are close to the ground truths.



\section{Application to Real-World}
We present an application to a real-world medical dataset.

{\bf Dataset.}
We take up a dataset ``Cholesterol Reduction'' \citep{Westfall2011} ({https://search.r-project.org/CRAN/refmans/multcomp/html/cholesterol.html}). 
This clinical study was conducted to assess the effect of three formulations of the same drug ($X$) on cholesterol reduction ($Y$).
The three formulations consisted of 20 mg taken once daily (“1 time”, $X=1$), 10 mg taken twice daily (“2 times”, $X=2$), and 5 mg taken four times daily (“4 times”, $X=4$).
The dataset 
has a sample size of 10 for each treatment group.
The purpose of the study was to determine which of the formulations is efficacious.
Previous studies examined multiple comparisons, specifically whether the conditions $\mathbb{E}[Y_4]>\mathbb{E}[Y_2]$, $\mathbb{E}[Y_2]>\mathbb{E}[Y_1]$, and $\mathbb{E}[Y_4]>\mathbb{E}[Y_1]$ simultaneously hold or not \citep{Westfall2011}.
The domain of $Y$ is bounded within $[5,19]$.

{We assume the exogeneity assumption holds, as \citet{Westfall2011} did not report any potential confounding factors.
We assume that the monotonicity assumption holds, meaning that increasing the formulation from 1 time to 2 times always leads to greater cholesterol reduction for all subjects, and similarly, increasing it from 2 times to 4 times consistently results in further cholesterol reduction for all subjects.}
We used the estimators described in Section~\ref{sec-exp}. 
We conduct the bootstrapping \citep{Efron1979} to 
provide the means and 95\% confidence intervals (CI) for each estimator.

{\bf Results on the moments of causal effects.}
First, we study the causal effects of changing the formulation from $X=1$ to $X=2$, 
i.e., $Y_2-Y_1$. The results are:
\begin{center}
\textbf{Mean}: $3.432$ (95\%CI: $[0.914,6.104]$).\\\vspace{0.1cm}
\textbf{Variance}: $3.072$ (95\%CI: $[0.297,8.610]$),\\\vspace{0.1cm}
\textbf{Standard deviation}: $1.753$ (95\%CI: $[0.545,2.934]$).\\\vspace{0.1cm}
\textbf{Skewness}: $21.027$ (95\%CI: $[-6.747,34.504]$),\\\vspace{0.1cm}
\textbf{Kurtosis}: $21.312$ (95\%CI: $[0.000,203.885]$).
\end{center}
Changing the treatment formulation from $X=1$ to $X=2$ increases the amount of cholesterol reduction.
Relatively large standard deviation  suggests that the causal effects exhibit some degree of heterogeneity.
The positive skewness suggests that the distribution of the causal effect may be positively skewed.
The kurtosis estimate has a large CI, and a larger sample size appears to be necessary to obtain more reliable estimates.

{Assume that we have enough samples such that these estimates are reliable. We have the average causal effect $3.432$. 
The relatively large standard deviation suggests that the causal effect exhibits a fair degree of heterogeneity across individuals. 
Our results exhibit substantially greater positive skewness than that of an exponential distribution, which has a skewness of $2$, and significantly higher kurtosis than that of a Gaussian distribution, which has a kurtosis of $3$. 
The large positive  skewness suggests that  there may be a larger number of individuals having effects smaller than the average $3.432$, rather than larger than $3.432$; and it suggests the existence of a small number of individuals who have effects that are significantly larger than the average, making the average greater than the median.
Finally, the large positive kurtosis value indicates a high number of outliers, which, given the large positive skewness, suggests a high number of individuals with causal effects that are significantly  higher than the average.}

The results for $Y_4-Y_1$ and $Y_4-Y_2$ are presented in Appendix \ref{appE},   along with the estimated bounds. 
The bounds 
are all relatively wide.

{\bf Results on the product moments of causal effects.}
Next, we study the covariance and correlation of causal effects $Y_2-Y_1$ and $Y_4-Y_2$, where the estimated mean of $Y_4-Y_2$ is  
$3.201$ (95\%CI: $[0.471,5.811]$). 
The estimates 
are
\begin{center}
\textbf{Covariance}: $-2.076$ (95\%CI: $[-6.235,0.297]$),\\\vspace{0.1cm}
\textbf{Correlation}: $-0.594$ (95\%CI: $[-1.000,0.750]$).
\end{center}
The results indicate that the causal effects $Y_2-Y_1$ and $Y_4-Y_2$ may be negatively correlated.

Consider the change in treatment from $X=1$ to 2 to 4. 
The mean of $Y_2-Y_1$ and $Y_4-Y_2$ is $3.432$ and $3.201$, respectively, indicating that the treatment changes from $X=1$ to $X=2$ and from $X=2$ to $X=4$ increase cholesterol reduction on average.
The correlation result offers a more detailed insight.
Patients who have larger causal effects $Y_2-Y_1$ than the average tend to have  causal effects $Y_4-Y_2$ smaller than the average.
Conversely, patients who have smaller causal effects $Y_2-Y_1$ than the average tend to have causal effects $Y_4-Y_2$ larger than the average.
This means that if one of the changes significantly increases cholesterol reduction, then the other change has a smaller effect. 
However, the estimate of the correlation has a large CI.
A larger sample size appears to be necessary to obtain a more reliable estimate.

Additionally, we show the estimated bounds of the covariance and correlation of causal effects $Y_2-Y_1$ and $Y_4-Y_2$ in Appendix \ref{appE}. 
The bounds 
are relatively wide.

\section{Conclusion}
This work moves beyond the  traditional focus on ACE to investigate the family of moments of causal effects. 
They could serve as fundamental tools for understanding the shape of the distribution of causal effects and their heterogeneity.
They can reveal asymmetries  in the distribution, suggesting the presence of distinct response groups that might require different treatment strategies.
This understanding is crucial for tailoring interventions and designing more effective policies. 

We have also studied 
conditional moments of causal effects,  
which  characterize the shape of the distribution of causal effects within a  subpopulation defined by subjects' covariates $W$. They provide complementary information to conditional ACE (CACE), i.e., $\mathbb{E}[Y_1-Y_0|W=w]$, by capturing higher-order properties beyond the mean.
The definition, identification, and bounds of the conditional moments of causal effects are discussed in Appendix~\ref{appD}.


{The exogeneity and monotonicity assumptions, while common in the causal inference literature, can  restrict the practical applicability of the identification results. The assumption of exogeneity, requiring the absence of unmeasured confounders, might be plausible in some applications, especially in certain controlled settings.The monotonicity assumption can be challenging to verify in practice. In such scenarios, the bounding results that relax the monotonicity assumption and depend primarily on exogeneity provide a range of plausible causal effects, still offering valuable information for guiding decisions. In general, these assumptions require a cautious interpretation of the findings in the practical use of the results.}

Future research will include sensitivity analysis and instrumental variable analysis for the  moments of causal effects under unmeasured confounders, extending methods developed for the ACE \citep{Wright1928, Cornfield1959, Imbens1994, VanderWeele2011, Ding2016}. 
Another direction will involve deriving narrower bounds.
{Deriving tighter bounds than those provided by the Fr\'{e}chet inequalities remains a highly challenging open mathematical problem when $m \geq 2$. 
In the finance area, some studies \citep{Lux2017,Bartl2017} provide improved Fr\'{e}chet–Hoeffding bounds by incorporating additional information to on the joint distribution.}
In the PoC literature, some studies \citep{Kuroki2011, Dawid2017} provide narrower bounds by incorporating additional information about third variables (e.g., covariates or mediators).



\section*{Acknowledgements}
The authors thank the anonymous reviewers for their time
and thoughtful comments.

\bibliography{aaai25}


\newpage
\appendix
\onecolumn

\section*{Appendix to ``Moments of Causal Effects"}

We provide the proofs of the lemmas and theorems in the body of the paper in Appendix \ref{appA}, the discussion of the central moments of causal effects in Appendix \ref{appB}, the discussion of the central product moments of causal effects in Appendix \ref{appC}, the discussion of the conditional moments of causal effects in Appendix \ref{appD}, the details and consistency of all estimators in the body of the paper in Appendix \ref{appCon}, and additional information about the application in the body of our paper in Appendix \ref{appE}.

\section{Proofs}
\label{appA}

In this appendix, we provide the proofs of the lemmas and theorems in the body of the paper.
We first show Fr\'{e}chet inequalities \citep{Frechet1935,Frechet1960}.
If $A_i$ are logical propositions or events, the Fréchet inequalities are
\begin{align}
\max\left\{\sum_{i=1}^n\mathbb{P}(A_i)-(n-1),0\right\} \leq \mathbb{P}\left(\bigwedge_{i=1}^n A_i\right) \leq \min_{i=1,\dots,n}\left\{\mathbb{P}(A_i)\right\},
\end{align}
where $\displaystyle \land$ is a logical conjunction.
Especially,
\begin{align}
\max\{\mathbb{P}(A)+\mathbb{P}(B)-1,0\}\leq \mathbb{P}(A \wedge B) \leq \min\{\mathbb{P}(A),\mathbb{P}(B)\}
\end{align}
holds.

{\bf Lemma \ref{lem1}.}
{\it
Under SCM ${\cal M}$, given $m\geq 1$, we have
\begin{align}
(Y_1-Y_0)^m&=(Y_1-Y_0)^m\mathbb{I}(Y_1>Y_0)+(-1)^m(Y_0-Y_1)^m\mathbb{I}(Y_0>Y_1)\nonumber\\
&=\int_{{\Omega_Y}^m} \mathbb{I}(Y_0<y_1\leq Y_1,Y_0<y_2\leq Y_1,\dots,Y_0<y_m\leq Y_1)dy_1\dots dy_m\nonumber\\
&+(-1)^m\int_{{\Omega_Y}^m} \mathbb{I}(Y_1<y_1\leq Y_0,Y_1<y_2\leq Y_0,\dots,Y_1<y_m\leq Y_0)dy_1\dots dy_m.
\end{align}
}

\begin{proof}
Given $m\geq 1$, we have
\begin{align}
&\int_{{\Omega_Y}^m} \mathbb{I}(Y_0<y_1\leq Y_1,Y_0<y_2\leq Y_1,\dots,Y_0<y_m\leq Y_1)dy_1\dots dy_m\nonumber\\
&=\int_{\Omega_Y}\mathbb{I}(Y_0<y_1\leq Y_1)dy_1\dots \int_{\Omega_Y}\mathbb{I}(Y_0<y_m\leq Y_1)dy_m\nonumber\\
&=(Y_1-Y_0)^m\mathbb{I}(Y_1>Y_0).
\end{align}
Similarly, given $m\geq 1$, we have
\begin{align}
&\int_{{\Omega_Y}^m} \mathbb{I}(Y_1<y_1\leq Y_0,Y_1<y_2\leq Y_0,\dots,Y_1<y_m\leq Y_0)dy_1\dots dy_m\nonumber\\
&=\int_{\Omega_Y}\mathbb{I}(Y_1<y_1\leq Y_0)dy_1\dots \int_{\Omega_Y}\mathbb{I}(Y_1<y_m\leq Y_0)dy_m\nonumber\\
&=(Y_0-Y_1)^m\mathbb{I}(Y_1>Y_0)^m\nonumber\\
&=(-1)^m(Y_1-Y_0)^m\mathbb{I}(Y_0>Y_1).
\end{align}
Then, given $m\geq 1$, we have
\begin{align}
(Y_1-Y_0)^m&=\int_{{\Omega_Y}^m} \mathbb{I}(Y_0<y_1\leq Y_1,Y_0<y_2\leq Y_1,\dots,Y_0<y_m\leq Y_1)dy_1\dots dy_m\nonumber\\
&+(-1)^m\int_{{\Omega_Y}^m} \mathbb{I}(Y_1<y_1\leq Y_0,Y_1<y_2\leq Y_0,\dots,Y_1<y_m\leq Y_0)dy_1\dots dy_m.
\end{align}
\end{proof}

{\bf Theorem \ref{theo1}.}
{\it
Under SCM ${\cal M}$ and Assumptions \ref{ASEXO2}, \ref{MONO2}, and \ref{exi1}, given $m\geq 1$, the $m$-th moment of causal effect $Y_1-Y_0$ is identifiable by $\sigma^{(m)}$, where
\begin{align}
\sigma^{(m)}&=\int_{{\Omega_Y}^m} \max\Big\{\min_{p=1,\dots,m}\{\mathbb{P}(Y<y_p|X=0)\}-\max_{p=1,\dots,m}\{\mathbb{P}(Y<y_p|X=1)\},0\Big\}dy_1\dots dy_m\nonumber\\
&+(-1)^m\int_{{\Omega_Y}^m} \max\Big\{\min_{p=1,\dots,m}\{\mathbb{P}(Y<y_p|X=1)\}-\max_{p=1,\dots,m}\{\mathbb{P}(Y<y_p|X=0)\},0\Big\}dy_1\dots dy_m.
\end{align}
}

\begin{proof}
$\mathbb{P}(Y_0<y_1\leq Y_1,Y_0<y_2\leq Y_1,\dots,Y_0<y_m\leq Y_1)$ and $\mathbb{P}(Y_1<y_1\leq Y_0,Y_1<y_2\leq Y_0,\dots,Y_1<y_m\leq Y_0)$ are identifiable by $\max\{\min_{p=1,\dots,m}\{\mathbb{P}(Y<y_p|X=0)\}-\max_{p=1,\dots,m}\{\mathbb{P}(Y<y_p|X=1),0\}$ and $\max\{\min_{p=1,\dots,m}\{\mathbb{P}(Y<y_p|X=1)\}-\max_{p=1,\dots,m}\{\mathbb{P}(Y<y_p|X=0),0\}$ respectively by Theorem 5.2 in \citep{Kawakami2024} under Assumptions \ref{ASEXO2} and \ref{MONO2}.
Then we have Eq.~\eqref{eq10}.
\end{proof}

{\bf Lemma \ref{lem3}.}
{\it
Under SCM ${\cal M}$ and Assumptions \ref{ASEXO2} and \ref{exi1}, given $m\geq 1$,
we have $l(y_1,\dots,y_m;i,j)\leq \mathbb{P}(Y_j<y_1\leq Y_i,Y_j<y_2\leq Y_i,\dots,Y_j<y_m\leq Y_i) \leq u(y_1,\dots,y_m;i,j)$
where 
\begin{align}
&l(y_1,\dots,y_m;i,j)=\max\left\{\sum_{p=1,\dots,m}\mathbb{P}(Y<y_p|X=j)-\sum_{p=1,\dots,m}\mathbb{P}(Y<y_p|X=i)-m+1,0\right\},
\end{align}
\begin{align}
&u(y_1,\dots,y_m;i,j)=\min\left\{\min_{p=1,\dots,m}\{\mathbb{P}(Y<y_p|X=j)\},1-\max_{p=1,\dots,m}\{\mathbb{P}(Y<y_p|X=i)\}\right\}
\end{align}
for $(i,j) \in \{(1,0),(0,1)\}$ and each $y_1, \dots, y_m \in \Omega_Y$.
}

\begin{proof}
From Fr\'{e}chet inequalities \citep{Frechet1935,Frechet1960}, given $m\geq 1$,
we have 
\begin{align}
&\mathbb{P}(Y_j<y_1\leq Y_i,Y_j<y_2\leq Y_i,\dots,Y_j<y_m\leq Y_i)\nonumber\\
&=\mathbb{P}(Y_j<y_1,y_1\leq Y_i,Y_j<y_2,y_2\leq Y_i,\dots,Y_j<y_m,y_m\leq Y_i)\nonumber\\
&\geq \max\left\{\sum_{p=1,\dots,m}\mathbb{P}(Y_j<y_p)+\sum_{p=1,\dots,m}\{1-\mathbb{P}(Y_i<y_p)\}-(2m-1),0\right\}\nonumber\\
&= \max\left\{\sum_{p=1,\dots,m}\mathbb{P}(Y_j<y_p)-\sum_{p=1,\dots,m}\mathbb{P}(Y_i<y_p)-(m-1),0\right\}\nonumber\\
&= \max\left\{\sum_{p=1,\dots,m}\mathbb{P}(Y<y_p|X=j)-\sum_{p=1,\dots,m}\mathbb{P}(Y<y_p|X=i)-(m-1),0\right\}
\end{align}
and
\begin{align}
&\mathbb{P}(Y_j<y_1\leq Y_i,Y_j<y_2\leq Y_i,\dots,Y_j<y_m\leq Y_i)\nonumber\\
&\leq\min\left\{\mathbb{P}(Y_j<y_1),\dots,\mathbb{P}(Y_j < y_m),\mathbb{P}(Y_i \leq y_1),\dots,\mathbb{P}(Y_i \leq y_m) \right\}\nonumber\\
&=\min\left\{\min_{p=1,\dots,m}\{\mathbb{P}(Y<y_p|X=j)\},1-\max_{p=1,\dots,m}\{\mathbb{P}(Y<y_p|X=i)\}\right\}
\end{align}
for $(i,j) \in \{(1,0),(0,0)\}$ and each $y_1, \dots,y_p \in \Omega_Y$.
\end{proof}

{\bf Theorem \ref{theo2}.}
{\it
Under SCM ${\cal M}$ and Assumptions \ref{ASEXO2} and \ref{exi1}, given $m\geq 1$, we have $\sigma_L^{(m)} \leq \mu^{(m)} \leq \sigma_U^{(m)}$, where

(A). When $m$ is an even number,
\begin{align}
&\sigma_L^{(m)}=\int_{{\Omega_Y}^m} l(y_1,\dots,y_m;1,0)dy_1\dots dy_m+\int_{{\Omega_Y}^m} l(y_1,\dots,y_m;0,1)dy_1\dots dy_m,
\end{align}
\begin{align}
&\sigma_U^{(m)}=\int_{{\Omega_Y}^m} u(y_1,\dots,y_m;1,0)dy_1\dots dy_m+\int_{{\Omega_Y}^m} u(y_1,\dots,y_m;0,1)dy_1\dots dy_m.
\end{align}

(B). When $m$ is an odd number,
\begin{align}
&\sigma_L^{(m)}=\int_{{\Omega_Y}^m} l(y_1,\dots,y_m;1,0)dy_1\dots dy_m-\int_{{\Omega_Y}^m} u(y_1,\dots,y_m;0,1)dy_1\dots dy_m,
\end{align}
\begin{align}
&\sigma_U^{(m)}=\int_{{\Omega_Y}^m} u(y_1,\dots,y_m;1,0)dy_1\dots dy_m-\int_{{\Omega_Y}^m} l(y_1,\dots,y_m;0,1)dy_1\dots dy_m.
\end{align}
}

\begin{proof}
When $m$ is an even number, we have
\begin{align}
\mu^{(m)}&=\int_{{\Omega_Y}^m} \mathbb{P}(Y_0<y_1\leq Y_1,Y_0<y_2\leq Y_1,\dots,Y_0<y_m\leq Y_1)dy_1\dots dy_m\nonumber\\
&+(-1)^m\int_{{\Omega_Y}^m} \mathbb{P}(Y_1<y_1\leq Y_0,Y_1<y_2\leq Y_0,\dots,Y_1<y_m\leq Y_0)dy_1\dots dy_m\nonumber\\
&=\int_{{\Omega_Y}^m} \mathbb{P}(Y_0<y_1\leq Y_1,Y_0<y_2\leq Y_1,\dots,Y_0<y_m\leq Y_1)dy_1\dots dy_m\nonumber\\
&+\int_{{\Omega_Y}^m} \mathbb{P}(Y_1<y_1\leq Y_0,Y_1<y_2\leq Y_0,\dots,Y_1<y_m\leq Y_0)dy_1\dots dy_m\nonumber\\
&\geq \int_{{\Omega_Y}^m} l(y_1,\dots,y_m;1,0)dy_1\dots dy_m+\int_{{\Omega_Y}^m} l(y_1,\dots,y_m;0,1)dy_1\dots dy_m,
\end{align}
and
\begin{align}
\mu^{(m)}&=\int_{{\Omega_Y}^m} \mathbb{P}(Y_0<y_1\leq Y_1,Y_0<y_2\leq Y_1,\dots,Y_0<y_m\leq Y_1)dy_1\dots dy_m\nonumber\\
&+(-1)^m\int_{{\Omega_Y}^m} \mathbb{P}(Y_1<y_1\leq Y_0,Y_1<y_2\leq Y_0,\dots,Y_1<y_m\leq Y_0)dy_1\dots dy_m\nonumber\\
&=\int_{{\Omega_Y}^m} \mathbb{P}(Y_0<y_1\leq Y_1,Y_0<y_2\leq Y_1,\dots,Y_0<y_m\leq Y_1)dy_1\dots dy_m\nonumber\\
&+\int_{{\Omega_Y}^m} \mathbb{P}(Y_1<y_1\leq Y_0,Y_1<y_2\leq Y_0,\dots,Y_1<y_m\leq Y_0)dy_1\dots dy_m\nonumber\\
&\leq \int_{{\Omega_Y}^m} u(y_1,\dots,y_m;1,0)dy_1\dots dy_m+\int_{{\Omega_Y}^m} u(y_1,\dots,y_m;0,1)dy_1\dots dy_m.
\end{align}

When $m$ is an even number, we have
\begin{align}
\mu^{(m)}&=\int_{{\Omega_Y}^m} \mathbb{P}(Y_0<y_1\leq Y_1,Y_0<y_2\leq Y_1,\dots,Y_0<y_m\leq Y_1)dy_1\dots dy_m\nonumber\\
&+(-1)^m\int_{{\Omega_Y}^m} \mathbb{P}(Y_1<y_1\leq Y_0,Y_1<y_2\leq Y_0,\dots,Y_1<y_m\leq Y_0)dy_1\dots dy_m\nonumber\\
&=\int_{{\Omega_Y}^m} \mathbb{P}(Y_0<y_1\leq Y_1,Y_0<y_2\leq Y_1,\dots,Y_0<y_m\leq Y_1)dy_1\dots dy_m\nonumber\\
&-\int_{{\Omega_Y}^m} \mathbb{P}(Y_1<y_1\leq Y_0,Y_1<y_2\leq Y_0,\dots,Y_1<y_m\leq Y_0)dy_1\dots dy_m\nonumber\\
&\geq \int_{{\Omega_Y}^m} l(y_1,\dots,y_m;1,0)dy_1\dots dy_m-\int_{{\Omega_Y}^m} u(y_1,\dots,y_m;0,1)dy_1\dots dy_m,
\end{align}
and
\begin{align}
\mu^{(m)}&=\int_{{\Omega_Y}^m} \mathbb{P}(Y_0<y_1\leq Y_1,Y_0<y_2\leq Y_1,\dots,Y_0<y_m\leq Y_1)dy_1\dots dy_m\nonumber\\
&+(-1)^m\int_{{\Omega_Y}^m} \mathbb{P}(Y_1<y_1\leq Y_0,Y_1<y_2\leq Y_0,\dots,Y_1<y_m\leq Y_0)dy_1\dots dy_m\nonumber\\
&=\int_{{\Omega_Y}^m} \mathbb{P}(Y_0<y_1\leq Y_1,Y_0<y_2\leq Y_1,\dots,Y_0<y_m\leq Y_1)dy_1\dots dy_m\nonumber\\
&-\int_{{\Omega_Y}^m} \mathbb{P}(Y_1<y_1\leq Y_0,Y_1<y_2\leq Y_0,\dots,Y_1<y_m\leq Y_0)dy_1\dots dy_m\nonumber\\
&\leq \int_{{\Omega_Y}^m} u(y_1,\dots,y_m;1,0)dy_1\dots dy_m-\int_{{\Omega_Y}^m} l(y_1,\dots,y_m;0,1)dy_1\dots dy_m.
\end{align}
\end{proof}

{\bf Lemma \ref{lem4}.}
{\it
Under SCM ${\cal M}$, for any $i,j,k,h \in \{1,\dots,R\}$, we have
\begin{align}
&(Y_i-Y_j)(Y_k-Y_h)\nonumber\\
&=(Y_i-Y_j)(Y_k-Y_h)\mathbb{I}(Y_i>Y_j,Y_k>Y_h)-(Y_j-Y_i)(Y_k-Y_h)\mathbb{I}(Y_j>Y_i,Y_k>Y_h)\nonumber\\
&-(Y_i-Y_j)(Y_h-Y_k)\mathbb{I}(Y_i>Y_j,Y_h>Y_k)+(Y_j-Y_i)(Y_h-Y_k)\mathbb{I}(Y_j>Y_i,Y_h>Y_k)\nonumber\\
&=\int_{\Omega_Y}\int_{\Omega_Y} \mathbb{I}(Y_j<y_1\leq Y_i,Y_h<y_2\leq Y_k)dy_1 dy_2-\int_{\Omega_Y}\int_{\Omega_Y} \mathbb{I}(Y_i<y_1\leq Y_j,Y_h<y_2\leq Y_k)dy_1 dy_2\nonumber\\
&-\int_{\Omega_Y}\int_{\Omega_Y} \mathbb{I}(Y_j<y_1\leq Y_i,Y_k<y_2\leq Y_h)dy_1 dy_2+\int_{\Omega_Y}\int_{\Omega_Y} \mathbb{I}(Y_i<y_1\leq Y_j,Y_k<y_2\leq Y_h)dy_1 dy_2.
\end{align}
}

\begin{proof}
For any $i,j,k,h \in \{1,\dots,R\}$, we have
\begin{align}
&\int_{\Omega_Y}\int_{\Omega_Y} \mathbb{I}(Y_j<y_1\leq Y_i,Y_h<y_2\leq Y_k)dy_1 dy_2-\int_{\Omega_Y}\int_{\Omega_Y} \mathbb{I}(Y_i<y_1\leq Y_j,Y_h<y_2\leq Y_k)dy_1 dy_2\nonumber\\
&-\int_{\Omega_Y}\int_{\Omega_Y} \mathbb{I}(Y_j<y_1\leq Y_i,Y_k<y_2\leq Y_h)dy_1 dy_2+\int_{\Omega_Y}\int_{\Omega_Y} \mathbb{I}(Y_i<y_1\leq Y_j,Y_k<y_2\leq Y_h)dy_1 dy_2\nonumber\\
&=\int_{\Omega_Y}\mathbb{I}(Y_j<y_1\leq Y_i)dy_1\int_{\Omega_Y} \mathbb{I}(Y_h<y_2\leq Y_k)dy_2-\int_{\Omega_Y}\mathbb{I}(Y_i<y_1\leq Y_j)dy_1\int_{\Omega_Y} \mathbb{I}(Y_h<y_2\leq Y_k)dy_2\nonumber\\
&-\int_{\Omega_Y}\mathbb{I}(Y_j<y_1\leq Y_i)dy_1\int_{\Omega_Y} \mathbb{I}(Y_k<y_2\leq Y_h)dy_2+\int_{\Omega_Y}\mathbb{I}(Y_i<y_1\leq Y_j)dy_1 \int_{\Omega_Y}\mathbb{I}(Y_k<y_2\leq Y_h)dy_2\nonumber\\
&=(Y_i-Y_j)(Y_k-Y_h)\mathbb{I}(Y_i>Y_j,Y_k>Y_h)-(Y_j-Y_i)(Y_k-Y_h)\mathbb{I}(Y_j>Y_i,Y_k>Y_h)\nonumber\\
&-(Y_i-Y_j)(Y_h-Y_k)\mathbb{I}(Y_i>Y_j,Y_h>Y_k)+(Y_j-Y_i)(Y_h-Y_k)\mathbb{I}(Y_j>Y_i,Y_h>Y_k)\nonumber\\
&=(Y_i-Y_j)(Y_k-Y_h)\mathbb{I}(Y_i>Y_j,Y_k>Y_h)+(Y_i-Y_j)(Y_k-Y_h)\mathbb{I}(Y_j>Y_i,Y_k>Y_h)\nonumber\\
&+(Y_i-Y_j)(Y_k-Y_h)\mathbb{I}(Y_i>Y_j,Y_h>Y_k)+(Y_i-Y_j)(Y_k-Y_h)\mathbb{I}(Y_j>Y_i,Y_h>Y_k)\nonumber\\
&=(Y_i-Y_j)(Y_k-Y_h).
\end{align}
\end{proof}

{\bf Theorem \ref{theo3}.}
{\it
Under SCM ${\cal M}$ and Assumptions \ref{ASEXO2}, \ref{MONO2}, and \ref{exi2}, for any $i,j,k,h \in \{1,\dots,R\}$, the product moment of $(Y_i-Y_j)$ and $(Y_k-Y_h)$ is identifiable by $\sigma(i,j;k,h)$, where
\begin{align}
&\sigma(i,j;k,h)\nonumber\\
&=\int_{{\Omega_Y}^2}\max\Big\{\min\{\mathbb{P}(Y<y_1|X=j),\mathbb{P}(Y<y_2|X=h)\}-\max\{\mathbb{P}(Y<y_1|X=i),\mathbb{P}(Y<y_2|X=k)\},0\Big\}dy_1dy_2\nonumber\\
&-\int_{{\Omega_Y}^2}\max\Big\{\min\{\mathbb{P}(Y<y_1|X=i),\mathbb{P}(Y<y_2|X=h)\}-\max\{\mathbb{P}(Y<y_1|X=j),\mathbb{P}(Y<y_2|X=k)\},0\Big\}dy_1dy_2\nonumber\\
&-\int_{{\Omega_Y}^2}\max\Big\{\min\{\mathbb{P}(Y<y_1|X=j),\mathbb{P}(Y<y_2|X=k)\}-\max\{\mathbb{P}(Y<y_1|X=i),\mathbb{P}(Y<y_2|X=h)\},0\Big\}dy_1dy_2\nonumber\\
&+\int_{{\Omega_Y}^2}\max\Big\{\min\{\mathbb{P}(Y<y_1|X=i),\mathbb{P}(Y<y_2|X=k)\}-\max\{\mathbb{P}(Y<y_1|X=j),\mathbb{P}(Y<y_2|X=h)\},0\Big\}dy_1dy_2.
\end{align}
}

\begin{proof}
$\mathbb{P}(Y_j<y_1\leq Y_i,Y_h<y_2\leq Y_k)$ are identifiable by $\max\{\min\{\mathbb{P}(Y<y_1|X=j),\mathbb{P}(Y<y_2|X=h)\}-\max\{\mathbb{P}(Y<y_1|X=i),\mathbb{P}(Y<y_2|X=k)\},0\}$ ($i,j,k,h \in \{1,\dots,R\}$) respectively by Theorem 5.2 in \citep{Kawakami2024} under Assumptions \ref{ASEXO2} and \ref{MONO2}.
Then, we have Eq.~\eqref{eq21}.
\end{proof}

{\bf Lemma \ref{lem6}.}
{\it
Under SCM ${\cal M}$ and Assumptions \ref{ASEXO2} and \ref{exi2}, for any $i,j,k,h \in \{1,\dots,R\}$,
we have $l(y_1,y_3;i,j,k,h)\leq \mathbb{P}(Y_j<y_1\leq Y_i,Y_h<y_2\leq Y_k)\leq u(y_1,y_3;i,j,k,h)$, where
\begin{align}
&l(y_1,y_2;i,j,k,h)\nonumber\\
&=\max\Big\{\mathbb{P}(Y<y_1|X=j)-\mathbb{P}(Y <y_1|X=i)+\mathbb{P}(Y<y_2|X=h)-\mathbb{P}(Y<y_2|X=k)-1,0\Big\},
\end{align}
\begin{align}
&u(y_1,y_2;i,j,k,h)\nonumber\\
&=\min\Big\{\min\{\mathbb{P}(Y<y_1|X=j),\mathbb{P}(Y<y_2|X=h)\},1-\max\{\mathbb{P}(Y<y_1|X=i),\mathbb{P}(Y<y_2|X=k)\}\Big\}
\end{align}
for each $i,j,k,h \in \{1,\dots,R\}$ and $y_1, y_2 \in \Omega_Y$.
}

\begin{proof}
From Fr\'{e}chet inequalities \citep{Frechet1935,Frechet1960},
for any $i,j,k,h \in \{1,\dots,R\}$,
we have 
\begin{align}
&\mathbb{P}(Y_j<y_1\leq Y_i,Y_h<y_2\leq Y_k)=\mathbb{P}(Y_j<y_1,y_1\leq Y_i,Y_h<y_2,y_2\leq Y_k)\nonumber\\
&\geq \max\Big\{\mathbb{P}(Y_j<y_1)+\mathbb{P}(y_1\leq Y_i)+\mathbb{P}(Y_h<y_2)+\mathbb{P}(y_2\leq Y_k)-3,0\Big\}\nonumber\\
&= \max\Big\{\mathbb{P}(Y_j<y_1)-\mathbb{P}(Y_i <y_1)+\mathbb{P}(Y_h<y_2)-\mathbb{P}(Y_k<y_2)-1,0\Big\}\nonumber\\
&= \max\Big\{\mathbb{P}(Y<y_1|X=j)-\mathbb{P}(Y <y_1|X=i)+\mathbb{P}(Y<y_2|X=h)-\mathbb{P}(Y<y_2|X=k)-1,0\Big\}
\end{align}
and 
\begin{align}
&\mathbb{P}(Y_j<y_1\leq Y_i,Y_h<y_2\leq Y_k)=\mathbb{P}(Y_j<y_1,y_1\leq Y_i,Y_h<y_2,y_2\leq Y_k)\nonumber\\
&\leq \min\Big\{\min\{\mathbb{P}(Y_j<y_1),\mathbb{P}(Y_h<y_2)\},1-\max\{\mathbb{P}(Y_i <y_1),\mathbb{P}(Y_k<y_2)\}\Big\}\nonumber\\
&=\min\Big\{\min\{\mathbb{P}(Y<y_1|X=j),\mathbb{P}(Y<y_2|X=h)\},1-\max\{\mathbb{P}(Y<y_1|X=i),\mathbb{P}(Y<y_2|X=k)\}\Big\}.
\end{align}
\end{proof}

{\bf Theorem \ref{theo4}.}
{\it
Under SCM ${\cal M}$ and Assumptions \ref{ASEXO2} and \ref{exi2}, for any $i,j,k,h \in \{1,\dots,R\}$, we have $\sigma_L(i,j;k,h) \leq \rho_{i,j;k,h}\leq \sigma_U(i,j;k,h)$, where
\begin{align}
\sigma_L(i,j;k,h)
&=\int_{{\Omega_Y}^2}l(y_1,y_2;i,j,k,h)dy_1dy_2-\int_{{\Omega_Y}^2}u(y_1,y_2;j,i,k,h)dy_1dy_2\nonumber\\
&\hspace{2cm}-\int_{{\Omega_Y}^2}u(y_1,y_2;i,j,h,k)dy_1dy_2+\int_{{\Omega_Y}^2}l(y_1,y_2;j,i,h,k)dy_1dy_2,
\end{align}
\begin{align}
\sigma_U(i,j;k,h)
&=\int_{{\Omega_Y}^2}u(y_1,y_2;i,j,k,h)dy_1dy_2-\int_{{\Omega_Y}^2}l(y_1,y_2;j,i,k,h)dy_1dy_2\nonumber\\
&\hspace{2cm}-\int_{{\Omega_Y}^2}l(y_1,y_2;i,j,h,k)dy_1dy_2+\int_{{\Omega_Y}^2}u(y_1,y_2;j,i,h,k)dy_1dy_2.
\end{align}}

\begin{proof}
For any $i,j,k,h \in \{1,\dots,R\}$,
we have
\begin{align}
\rho_{i,j;k,h}
&=\int_{\Omega_Y}\int_{\Omega_Y} \mathbb{P}(Y_j<y_1\leq Y_i,Y_h<y_2\leq Y_k)dy_1 dy_2-\int_{\Omega_Y}\int_{\Omega_Y} \mathbb{P}(Y_i<y_1\leq Y_j,Y_h<y_2\leq Y_k)dy_1 dy_2\nonumber\\
&-\int_{\Omega_Y}\int_{\Omega_Y} \mathbb{P}(Y_j<y_1\leq Y_i,Y_k<y_2\leq Y_h)dy_1 dy_2+\int_{\Omega_Y}\int_{\Omega_Y} \mathbb{P}(Y_i<y_1\leq Y_j,Y_k<y_2\leq Y_h)dy_1 dy_2\nonumber\\
&\geq \int_{{\Omega_Y}^2}l(y_1,y_2;i,j,k,h)dy_1dy_2-\int_{{\Omega_Y}^2}u(y_1,y_2;j,i,k,h)dy_1dy_2\nonumber\\
&\hspace{4cm}-\int_{{\Omega_Y}^2}u(y_1,y_2;i,j,h,k)dy_1dy_2+\int_{{\Omega_Y}^2}l(y_1,y_2;j,i,h,k)dy_1dy_2,
\end{align}
and
\begin{align}
\rho_{i,j;k,h}
&=\int_{\Omega_Y}\int_{\Omega_Y} \mathbb{P}(Y_j<y_1\leq Y_i,Y_h<y_2\leq Y_k)dy_1 dy_2-\int_{\Omega_Y}\int_{\Omega_Y} \mathbb{P}(Y_i<y_1\leq Y_j,Y_h<y_2\leq Y_k)dy_1 dy_2\nonumber\\
&-\int_{\Omega_Y}\int_{\Omega_Y} \mathbb{P}(Y_j<y_1\leq Y_i,Y_k<y_2\leq Y_h)dy_1 dy_2+\int_{\Omega_Y}\int_{\Omega_Y} \mathbb{P}(Y_i<y_1\leq Y_j,Y_k<y_2\leq Y_h)dy_1 dy_2\nonumber\\
&\leq \int_{{\Omega_Y}^2}u(y_1,y_2;i,j,k,h)dy_1dy_2-\int_{{\Omega_Y}^2}l(y_1,y_2;j,i,k,h)dy_1dy_2\nonumber\\
&\hspace{4cm}-\int_{{\Omega_Y}^2}l(y_1,y_2;i,j,h,k)dy_1dy_2+\int_{{\Omega_Y}^2}u(y_1,y_2;j,i,h,k)dy_1dy_2.
\end{align}
\end{proof}

\section{Identification and Bounds of the Central Moments of Causal Effects}
\label{appB}

In this section, we discuss the central moments of causal effects.

We make the following assumption:
\begin{assumption}[Existence of integrals]
\label{exi3}
Under SCM ${\cal M}$, given $m\geq 1$, 
$\overline{\mu}^{(m)}<\infty$ and
$\int_{{\Omega_Y}^m} \mathbb{P}(Y_i-\mathbb{E}[Y_i]<y_1\leq Y_j-\mathbb{E}[Y_j],Y_i-\mathbb{E}[Y_i]<y_2\leq Y_j-\mathbb{E}[Y_j],\dots,Y_i-\mathbb{E}[Y_i]<y_m\leq Y_j-\mathbb{E}[Y_j])dy_1\dots dy_m<\infty$ hold for $(i,j)=\{(0,1),(1,0)\}$.
\end{assumption}

{\bf Identification of the central moment of causal effects.}

Then, we have the following identification theorem.
\begin{theorem}[Identification of the central moment of causal effects]
Under SCM ${\cal M}$ and Assumptions \ref{ASEXO2}, \ref{MONO2}, and \ref{exi3}, given $m\geq 1$, $\overline{\mu}^{(m)}$ is identifiable by $\bar{\sigma}^{(m)}$, where
\begin{align}
\label{eq91}
\bar{\sigma}^{(m)}&=\int_{{\Omega_Y}^m} \max\Big\{\min_{p=1,\dots,m}\{\mathbb{P}(Y-\mathbb{E}[Y|X=0]<y_p|X=0)\}\nonumber\\
&\hspace{5cm}-\max_{p=1,\dots,m}\{\mathbb{P}(Y-\mathbb{E}[Y|X=1]<y_p|X=1)\},0\Big\}dy_1\dots dy_m\nonumber\\
&+(-1)^m\int_{{\Omega_Y}^m}\max\Big\{\min_{p=1,\dots,m}\{\mathbb{P}(Y-\mathbb{E}[Y|X=1]<y_p|X=1)\}\nonumber\\
&\hspace{5cm}-\max_{p=1,\dots,m}\{\mathbb{P}(Y-\mathbb{E}[Y|X=0]<y_p|X=0)\},0\Big\}dy_1\dots dy_m.
\end{align}
\end{theorem}

\begin{proof}
Given $m\geq 1$, we have
\begin{align}
&\int_{{\Omega_Y}^m} \mathbb{I}(Y_0-\mathbb{E}[Y_0]<y_1\leq Y_1-\mathbb{E}[Y_1],Y_0-\mathbb{E}[Y_0]<y_2\leq Y_1-\mathbb{E}[Y_1],\dots,\nonumber\\
&\hspace{7cm}Y_0-\mathbb{E}[Y_0]<y_m\leq Y_1-\mathbb{E}[Y_1])dy_1\dots dy_m\nonumber\\
&=\int_{\Omega_Y}\mathbb{I}(Y_0-\mathbb{E}[Y_0]<y_1\leq Y_1-\mathbb{E}[Y_1])dy_1\dots \int_{\Omega_Y}\mathbb{I}(Y_0-\mathbb{E}[Y_0]<y_m\leq Y_1-\mathbb{E}[Y_1])dy_m\nonumber\\
&=\{(Y_1-Y_0)-(\mathbb{E}[Y_1]-\mathbb{E}[Y_0])\}^m\mathbb{I}(Y_1-\mathbb{E}[Y_1]>Y_0-\mathbb{E}[Y_0]).
\end{align}
Similarly, given $m\geq 1$, we have
\begin{align}
&\int_{{\Omega_Y}^m} \mathbb{I}(Y_1-\mathbb{E}[Y_1]<y_1\leq Y_0-\mathbb{E}[Y_1],Y_1-\mathbb{E}[Y_1]<y_2\leq Y_0-\mathbb{E}[Y_0],\dots,\nonumber\\
&\hspace{7cm}Y_1-\mathbb{E}[Y_1]<y_m\leq Y_0-\mathbb{E}[Y_0])dy_1\dots dy_m\nonumber\\
&=\int_{\Omega_Y}\mathbb{I}(Y_1-\mathbb{E}[Y_1]<y_1\leq Y_0-\mathbb{E}[Y_0])dy_1\dots \int_{\Omega_Y}\mathbb{I}(Y_1-\mathbb{E}[Y_1]<y_m\leq Y_0-\mathbb{E}[Y_0])dy_m\nonumber\\
&=\{(Y_0-Y_1)-(\mathbb{E}[Y_0]-\mathbb{E}[Y_1])\}^m\mathbb{I}(Y_0-\mathbb{E}[Y_0]>Y_1-\mathbb{E}[Y_1])^m\nonumber\\
&=(-1)^m\{(Y_1-Y_0)-(\mathbb{E}[Y_1]-\mathbb{E}[Y_0])\}^m\mathbb{I}(Y_0-\mathbb{E}[Y_0]>Y_1-\mathbb{E}[Y_1]).
\end{align}
Then, given $m\geq 1$, we have
\begin{align}
&\{(Y_1-Y_0)-(\mathbb{E}[Y_1]-\mathbb{E}[Y_0])\}^m\nonumber\\
&=\int_{{\Omega_Y}^m} \mathbb{I}(Y_0-\mathbb{E}[Y_0]<y_1\leq Y_1-\mathbb{E}[Y_1],Y_0-\mathbb{E}[Y_0]<y_2\leq Y_1-\mathbb{E}[Y_1],\dots,\nonumber\\
&\hspace{5cm}Y_0-\mathbb{E}[Y_0]<y_m\leq Y_1-\mathbb{E}[Y_1])dy_1\dots dy_m\nonumber\\
&+(-1)^m\int_{{\Omega_Y}^m} \mathbb{I}(Y_1-\mathbb{E}[Y_1]<y_1\leq Y_0-\mathbb{E}[Y_0],Y_1-\mathbb{E}[Y_1]<y_2\leq Y_0-\mathbb{E}[Y_0],\dots,\nonumber\\
&\hspace{5cm}Y_1-\mathbb{E}[Y_1]<y_m\leq Y_0-\mathbb{E}[Y_0])dy_1\dots dy_m.
\end{align}
Taking expectations, given $m\geq 1$, we obtain
\begin{align}
&\overline{\mu}^{(m)}=\int_{{\Omega_Y}^m} \mathbb{P}(Y_0-\mathbb{E}[Y_0]<y_1\leq Y_1-\mathbb{E}[Y_1],Y_0-\mathbb{E}[Y_0]<y_2\leq Y_1-\mathbb{E}[Y_1],\dots,\nonumber\\
&\hspace{7cm}Y_0-\mathbb{E}[Y_0]<y_m\leq Y_1-\mathbb{E}[Y_1])dy_1\dots dy_m\nonumber\\
&+(-1)^m\int_{{\Omega_Y}^m} \mathbb{P}(Y_1-\mathbb{E}[Y_1]<y_1\leq Y_0-\mathbb{E}[Y_0],Y_1-\mathbb{E}[Y_1]<y_2\leq Y_0-\mathbb{E}[Y_0],\dots,\nonumber\\
&\hspace{7cm}Y_1-\mathbb{E}[Y_1]<y_m\leq Y_0-\mathbb{E}[Y_0])dy_1\dots dy_m.
\end{align}
Since $\mathbb{P}(Y_0-\mathbb{E}[Y_0]<y_1\leq Y_1-\mathbb{E}[Y_1],Y_0-\mathbb{E}[Y_0]<y_2\leq Y_1-\mathbb{E}[Y_1],\dots,Y_0-\mathbb{E}[Y_0]<y_m\leq Y_1-\mathbb{E}[Y_1])$ and $\mathbb{P}(Y_1-\mathbb{E}[Y_1]<y_1\leq Y_0-\mathbb{E}[Y_0],Y_1-\mathbb{E}[Y_1]<y_2\leq Y_0-\mathbb{E}[Y_0],\dots,Y_1-\mathbb{E}[Y_1]<y_m\leq Y_0-\mathbb{E}[Y_0])$ are identifiable by re-writing $f_Y$ as $f(x,U_Y)=Y_x-\mathbb{E}[Y_x]$ in Theorem 5.2 in \citep{Kawakami2024}, we have, given $m\geq 1$,
\begin{align}
&\overline{\mu}^{(m)}=\int_{{\Omega_Y}^m} \max\Big\{\min_{p=1,\dots,m}\{\mathbb{P}(Y-\mathbb{E}[Y|X=0]<y_p|X=0)\}\nonumber\\
&\hspace{5cm}-\max_{p=1,\dots,m}\{\mathbb{P}(Y-\mathbb{E}[Y|X=1]<y_p|X=1)\},0\Big\}dy_1\dots dy_m\nonumber\\
&+(-1)^m\int_{{\Omega_Y}^m}\max\Big\{\min_{p=1,\dots,m}\{\mathbb{P}(Y-\mathbb{E}[Y|X=1]<y_p|X=1)\}\nonumber\\
&\hspace{5cm}-\max_{p=1,\dots,m}\{\mathbb{P}(Y-\mathbb{E}[Y|X=0]<y_p|X=0)\},0\Big\}dy_1\dots dy_m.
\end{align}
\end{proof}
Moments of causal effects are expressed as the combination of conditional CDFs and expectations.

{\bf Bounding the central moment of causal effects.}

Assumption \ref{MONO2} may sometimes be considered implausible by researchers. Therefore, we derive bounds for the central moments of causal effects without relying on Assumption \ref{MONO2}.

\begin{lemma}

Under SCM ${\cal M}$ and Assumptions \ref{ASEXO2} and \ref{exi3}, given $m\geq 1$,
we have $\bar{l}(y_1,\dots,y_m;i,j)\leq \mathbb{P}(Y_j-\mathbb{E}[Y_j]<y_1\leq Y_i-\mathbb{E}[Y_i],Y_j-\mathbb{E}[Y_j]<y_2\leq Y_i-\mathbb{E}[Y_i],\dots,Y_j-\mathbb{E}[Y_j]<y_m\leq Y_i-\mathbb{E}[Y_i]) \leq \bar{u}(y_1,\dots,y_m;i,j)$,
where 
\begin{align}
&\bar{l}(y_1,\dots,y_m;i,j)=\nonumber\\
&\max\left\{\sum_{p=1,\dots,m}\mathbb{P}(Y-\mathbb{E}[Y|X=j]<y_p|X=j)-\sum_{p=1,\dots,m}\mathbb{P}(Y-\mathbb{E}[Y|X=i]<y_p|X=i)-(m-1),0\right\}    
\end{align}
and 
\begin{align}
&\bar{u}(y_1,\dots,y_m;i,j)=\nonumber\\
&\min\left\{\min_{p=1,\dots,m}\{\mathbb{P}(Y-\mathbb{E}[Y|X=j]<y_p|X=j)\},1-\max_{p=1,\dots,m}\{\mathbb{P}(Y-\mathbb{E}[Y|X=i]<y_p|X=i)\}\right\}
\end{align}
for $(i,j) \in \{(1,0),(0,1)\}$ and each $y_1, \dots, y_m \in \Omega_Y$.
\end{lemma}

\begin{proof}
From Fr\'{e}chet inequalities \citep{Frechet1935,Frechet1960}, given $m\geq 1$,
we have 
\begin{align}
&\mathbb{P}(Y_j-\mathbb{E}[Y_j]<y_1\leq Y_i-\mathbb{E}[Y_i],Y_j-\mathbb{E}[Y_j]<y_2\leq Y_i-\mathbb{E}[Y_i],\dots,Y_j-\mathbb{E}[Y_j]<y_m\leq Y_i-\mathbb{E}[Y_i])\nonumber\\
&=\mathbb{P}(Y_j-\mathbb{E}[Y_j]<y_1,y_1\leq Y_i-\mathbb{E}[Y_i],Y_j-\mathbb{E}[Y_j]<y_2,y_2\leq Y_i-\mathbb{E}[Y_i],\dots,Y_j-\mathbb{E}[Y_j]<y_m,y_m\leq Y_i-\mathbb{E}[Y_i])\nonumber\\
&\geq \max\left\{\sum_{p=1,\dots,m}\mathbb{P}(Y_j-\mathbb{E}[Y_j]<y_p)+\sum_{p=1,\dots,m}\{1-\mathbb{P}(Y_i-\mathbb{E}[Y_i]<y_p)\}-(2m-1),0\right\}\nonumber\\
&= \max\left\{\sum_{p=1,\dots,m}\mathbb{P}(Y_j-\mathbb{E}[Y_j]<y_p)-\sum_{p=1,\dots,m}\mathbb{P}(Y_i-\mathbb{E}[Y_i]<y_p)-(m-1),0\right\}\nonumber\\
&= \max\Bigg\{\sum_{p=1,\dots,m}\mathbb{P}(Y-\mathbb{E}[Y|X=j]<y_p|X=j)-\sum_{p=1,\dots,m}\mathbb{P}(Y-\mathbb{E}[Y|X=i]<y_p|X=i)-(m-1),0\Bigg\}
\end{align}
and
\begin{align}
&\mathbb{P}(Y_j-\mathbb{E}[Y_j]<y_1\leq Y_i-\mathbb{E}[Y_i],Y_j-\mathbb{E}[Y_j]<y_2\leq Y_i-\mathbb{E}[Y_i],\dots,Y_j-\mathbb{E}[Y_j]<y_m\leq Y_i-\mathbb{E}[Y_i])\nonumber\\
&\leq\min\left\{\mathbb{P}(Y_j-\mathbb{E}[Y_j]<y_1),\dots,\mathbb{P}(Y_j-\mathbb{E}[Y_j] < y_m),\mathbb{P}(Y_i-\mathbb{E}[Y_i] \leq y_1),\dots,\mathbb{P}(Y_i -\mathbb{E}[Y_i]\leq y_m) \right\}\nonumber\\
&=\min\left\{\min_{p=1,\dots,m}\{\mathbb{P}(Y-\mathbb{E}[Y|X=j]<y_p|X=j)\},1-\max_{p=1,\dots,m}\{\mathbb{P}(Y-\mathbb{E}[Y|X=i]<y_p|X=i)\}\right\}
\end{align}
for $(i,j) \in \{(1,0),(0,0)\}$ and each $y_1, \dots,y_p \in \Omega_Y$.    
\end{proof}

\begin{theorem}[Bounds of the central moment of causal effects]
Under SCM ${\cal M}$ and Assumptions \ref{ASEXO2} and \ref{exi3}, given $m\geq 1$, we have $\bar{\sigma}_L^{(m)} \leq \overline{\mu}^{(m)} \leq \bar{\sigma}_U^{(m)}$, where

(A). When $m$ is an even number, 
\begin{align}
\label{eq132}
&\bar{\sigma}_L^{(m)}=\int_{{\Omega_Y}^m} \bar{l}(y_1,\dots,y_m;1,0)dy_1\dots dy_m+\int_{{\Omega_Y}^m} \bar{l}(y_1,\dots,y_m;0,1)dy_1\dots dy_m,
\end{align}
\begin{align}
&\bar{\sigma}_U^{(m)}=\int_{{\Omega_Y}^m} \bar{u}(y_1,\dots,y_m;1,0)dy_1\dots dy_m+\int_{{\Omega_Y}^m} \bar{u}(y_1,\dots,y_m;0,1)dy_1\dots dy_m.
\end{align}

(B). When $m$ is an odd number,
\begin{align}
&\bar{\sigma}_L^{(m)}=\int_{{\Omega_Y}^m} \bar{l}(y_1,\dots,y_m;1,0)dy_1\dots dy_m-\int_{{\Omega_Y}^m} \bar{u}(y_1,\dots,y_m;0,1)dy_1\dots dy_m,
\end{align}
\begin{align}
\label{eq135}
&\bar{\sigma}_U^{(m)}=\int_{{\Omega_Y}^m} \bar{u}(y_1,\dots,y_m;1,0)dy_1\dots dy_m-\int_{{\Omega_Y}^m} \bar{l}(y_1,\dots,y_m;0,1)dy_1\dots dy_m.
\end{align}
\end{theorem}

\begin{proof}    
When $m$ is an even number, we have
\begin{align}
&\overline{\mu}^{(m)}=\int_{{\Omega_Y}^m} \mathbb{P}(Y_0-\mathbb{E}[Y_0]<y_1\leq Y_1-\mathbb{E}[Y_1],Y_0-\mathbb{E}[Y_0]<y_2\leq Y_1-\mathbb{E}[Y_1],\dots,\nonumber\\
&\hspace{8cm}Y_0-\mathbb{E}[Y_0]<y_m\leq Y_1-\mathbb{E}[Y_1])dy_1\dots dy_m\nonumber\\
&+(-1)^m\int_{{\Omega_Y}^m} \mathbb{P}(Y_1-\mathbb{E}[Y_1]<y_1\leq Y_0-\mathbb{E}[Y_0],Y_1-\mathbb{E}[Y_1]<y_2\leq Y_0-\mathbb{E}[Y_0],\dots,\nonumber\\
&\hspace{8cm}Y_1-\mathbb{E}[Y_1]<y_m\leq Y_0-\mathbb{E}[Y_0])dy_1\dots dy_m\nonumber\\
&=\int_{{\Omega_Y}^m} \mathbb{P}(Y_0-\mathbb{E}[Y_0]<y_1\leq Y_1-\mathbb{E}[Y_1],Y_0-\mathbb{E}[Y_0]<y_2\leq Y_1-\mathbb{E}[Y_1],\dots,\nonumber\\
&\hspace{8cm}Y_0-\mathbb{E}[Y_0]<y_m\leq Y_1-\mathbb{E}[Y_1])dy_1\dots dy_m\nonumber\\
&+\int_{{\Omega_Y}^m} \mathbb{P}(Y_1-\mathbb{E}[Y_1]<y_1\leq Y_0-\mathbb{E}[Y_0],Y_1-\mathbb{E}[Y_1-\mathbb{E}[Y_0]]<y_2\leq Y_0-\mathbb{E}[Y_0],\dots,\nonumber\\
&\hspace{8cm}Y_1-\mathbb{E}[Y_1]<y_m\leq Y_0-\mathbb{E}[Y_0])dy_1\dots dy_m\nonumber\\
&\geq \int_{{\Omega_Y}^m} \bar{l}(y_1,\dots,y_m;1,0)dy_1\dots dy_m+\int_{{\Omega_Y}^m} \bar{l}(y_1,\dots,y_m;0,1)dy_1\dots dy_m
\end{align}
and
\begin{align}
&\overline{\mu}^{(m)}=\int_{{\Omega_Y}^m} \mathbb{P}(Y_0-\mathbb{E}[Y_0]<y_1\leq Y_1-\mathbb{E}[Y_1],Y_0-\mathbb{E}[Y_0]<y_2\leq Y_1-\mathbb{E}[Y_1],\dots,\nonumber\\
&\hspace{8cm}Y_0-\mathbb{E}[Y_0]<y_m\leq Y_1-\mathbb{E}[Y_1])dy_1\dots dy_m\nonumber\\
&+(-1)^m\int_{{\Omega_Y}^m} \mathbb{P}(Y_1-\mathbb{E}[Y_1]<y_1\leq Y_0-\mathbb{E}[Y_0],Y_1-\mathbb{E}[Y_1]<y_2\leq Y_0-\mathbb{E}[Y_0],\dots,\nonumber\\
&\hspace{8cm}Y_1-\mathbb{E}[Y_1]<y_m\leq Y_0-\mathbb{E}[Y_0])dy_1\dots dy_m\nonumber\\
&=\int_{{\Omega_Y}^m} \mathbb{P}(Y_0-\mathbb{E}[Y_0]<y_1\leq Y_1-\mathbb{E}[Y_1],Y_0-\mathbb{E}[Y_0]<y_2\leq Y_1-\mathbb{E}[Y_1],\dots,\nonumber\\
&\hspace{8cm}Y_0-\mathbb{E}[Y_0]<y_m\leq Y_1-\mathbb{E}[Y_1])dy_1\dots dy_m\nonumber\\
&+\int_{{\Omega_Y}^m} \mathbb{P}(Y_1-\mathbb{E}[Y_1]<y_1\leq Y_0-\mathbb{E}[Y_0],Y_1-\mathbb{E}[Y_1]<y_2\leq Y_0-\mathbb{E}[Y_0],\dots,\nonumber\\
&\hspace{8cm}Y_1-\mathbb{E}[Y_1]<y_m\leq Y_0-\mathbb{E}[Y_0])dy_1\dots dy_m\nonumber\\
&\leq \int_{{\Omega_Y}^m} \bar{u}(y_1,\dots,y_m;1,0)dy_1\dots dy_m+\int_{{\Omega_Y}^m} \bar{u}(y_1,\dots,y_m;0,1)dy_1\dots dy_m.
\end{align}

When $m$ is an even number, we have
\begin{align}
&\overline{\mu}^{(m)}=\int_{{\Omega_Y}^m} \mathbb{P}(Y_0-\mathbb{E}[Y_0]<y_1\leq Y_1-\mathbb{E}[Y_1],Y_0-\mathbb{E}[Y_0]<y_2\leq Y_1-\mathbb{E}[Y_1],\dots,\nonumber\\
&\hspace{8cm}Y_0-\mathbb{E}[Y_0]<y_m\leq Y_1-\mathbb{E}[Y_1])dy_1\dots dy_m\nonumber\\
&+(-1)^m\int_{{\Omega_Y}^m} \mathbb{P}(Y_1-\mathbb{E}[Y_1]<y_1\leq Y_0-\mathbb{E}[Y_0],Y_1-\mathbb{E}[Y_1]<y_2\leq Y_0-\mathbb{E}[Y_0],\dots,\nonumber\\
&\hspace{8cm}Y_1-\mathbb{E}[Y_1]<y_m\leq Y_0-\mathbb{E}[Y_0])dy_1\dots dy_m\nonumber\\
&=\int_{{\Omega_Y}^m} \mathbb{P}(Y_0-\mathbb{E}[Y_0]<y_1\leq Y_1-\mathbb{E}[Y_1],Y_0-\mathbb{E}[Y_0]<y_2\leq Y_1-\mathbb{E}[Y_1],\dots,\nonumber\\
&\hspace{8cm}Y_0-\mathbb{E}[Y_0]<y_m\leq Y_1-\mathbb{E}[Y_1])dy_1\dots dy_m\nonumber\\
&-\int_{{\Omega_Y}^m} \mathbb{P}(Y_1-\mathbb{E}[Y_1]<y_1\leq Y_0-\mathbb{E}[Y_0],Y_1-\mathbb{E}[Y_1]<y_2\leq Y_0-\mathbb{E}[Y_0],\dots,\nonumber\\
&\hspace{8cm}Y_1-\mathbb{E}[Y_1]<y_m\leq Y_0-\mathbb{E}[Y_0])dy_1\dots dy_m\nonumber\\
&\geq \int_{{\Omega_Y}^m} \bar{l}(y_1,\dots,y_m;1,0)dy_1\dots dy_m-\int_{{\Omega_Y}^m} \bar{u}(y_1,\dots,y_m;0,1)dy_1\dots dy_m,
\end{align}
and
\begin{align}
&\overline{\mu}^{(m)}=\int_{{\Omega_Y}^m} \mathbb{P}(Y_0-\mathbb{E}[Y_0]<y_1\leq Y_1-\mathbb{E}[Y_1],Y_0-\mathbb{E}[Y_0]<y_2\leq Y_1-\mathbb{E}[Y_1],\dots,\nonumber\\
&\hspace{8cm}Y_0-\mathbb{E}[Y_0]<y_m\leq Y_1-\mathbb{E}[Y_1])dy_1\dots dy_m\nonumber\\
&+(-1)^m\int_{{\Omega_Y}^m} \mathbb{P}(Y_1-\mathbb{E}[Y_1]<y_1\leq Y_0-\mathbb{E}[Y_0],Y_1-\mathbb{E}[Y_1]<y_2\leq Y_0-\mathbb{E}[Y_0],\dots,\nonumber\\
&\hspace{8cm}Y_1-\mathbb{E}[Y_1]<y_m\leq Y_0-\mathbb{E}[Y_0])dy_1\dots dy_m\nonumber\\
&=\int_{{\Omega_Y}^m} \mathbb{P}(Y_0-\mathbb{E}[Y_0]<y_1\leq Y_1-\mathbb{E}[Y_1],Y_0-\mathbb{E}[Y_0]<y_2\leq Y_1-\mathbb{E}[Y_1],\dots,\nonumber\\
&\hspace{8cm}Y_0-\mathbb{E}[Y_0]<y_m\leq Y_1-\mathbb{E}[Y_1])dy_1\dots dy_m\nonumber\\
&-\int_{{\Omega_Y}^m} \mathbb{P}(Y_1-\mathbb{E}[Y_1]<y_1\leq Y_0-\mathbb{E}[Y_0],Y_1-\mathbb{E}[Y_1]<y_2\leq Y_0-\mathbb{E}[Y_0],\dots,\nonumber\\
&\hspace{8cm}Y_1-\mathbb{E}[Y_1]<y_m\leq Y_0-\mathbb{E}[Y_0])dy_1\dots dy_m\nonumber\\
&\leq \int_{{\Omega_Y}^m} \bar{u}(y_1,\dots,y_m;1,0)dy_1\dots dy_m-\int_{{\Omega_Y}^m} \bar{l}(y_1,\dots,y_m;0,1)dy_1\dots dy_m.
\end{align}
\end{proof}
We can calculate the bounds of the skewness as follows:
\begin{align}
&\frac{\bar{\sigma}_L^{(3)}}{\bar{\sigma}_U^{(2)\frac{3}{2}}}\mathbb{I}(\bar{\sigma}_L^{(3)}\geq 0)+\frac{\bar{\sigma}_L^{(3)}}{\bar{\sigma}_L^{(2)\frac{3}{2}}}\mathbb{I}(\bar{\sigma}_L^{(3)}< 0)\leq \frac{\overline{\mu}^{(3)}}{{\overline{\mu}^{(2)}}^{\frac{3}{2}}} \leq \frac{\bar{\sigma}_U^{(3)}}{\bar{\sigma}_L^{(2)\frac{3}{2}}}\mathbb{I}(\bar{\sigma}_U^{(3)}\geq 0)+\frac{\bar{\sigma}_U^{(3)}}{\bar{\sigma}_U^{(2)\frac{3}{2}}}\mathbb{I}(\bar{\sigma}_U^{(3)}< 0),
\end{align}
and the bounds of the kurtosis as follows:
\begin{align}
&\frac{\bar{\sigma}_L^{(4)}}{\bar{\sigma}_U^{(2)2}}\leq \frac{\overline{\mu}^{(4)}}{{\overline{\mu}^{(2)}}^{2}} \leq\frac{\bar{\sigma}_U^{(4)}}{\bar{\sigma}_L^{(2)2}}.
\end{align}

\section{Identification and Bounds of the Central Product Moments of Causal Effects}
\label{appC}

In this section, we discuss the central product moments of causal effects.
We make the following assumption:
\begin{assumption}[Existence of integrals]
\label{exi4}
Under SCM ${\cal M}$, given $m\geq 1$,
$\overline{\rho}_{i,j;k,h}<\infty$ and
$\int_{\Omega_Y}\int_{\Omega_Y} \mathbb{P}(Y_j-\mathbb{E}[Y_j]<y_1\leq Y_i-\mathbb{E}[Y_i],Y_h-\mathbb{E}[Y_h]<y_2\leq Y_k-\mathbb{E}[Y_k])dy_1 dy_2<\infty$ hold for any $i,j,k,h \in \{1,\dots,R\}$.
\end{assumption}

{\bf Identification of the central product moment of causal effects.}

\begin{theorem}[Identification of the central product moment of causal effects]
Under SCM ${\cal M}$ and Assumptions \ref{ASEXO2}, \ref{MONO2}, and \ref{exi4}, for any $i,j,k,h \in \{1,\dots,R\}$, $\overline{\rho}_{i,j;k,h}$ is identifiable by $\bar{\sigma}(i,j;k,h)$, where
\begin{align}
\label{eq65}
\bar{\sigma}(i,j;k,h)
&=\int_{{\Omega_Y}^2}\max\Big\{\min\{\mathbb{P}(Y-\mathbb{E}[Y|X=j]<y_1|X=j),\mathbb{P}(Y-\mathbb{E}[Y|X=h]<y_2|X=h)\}\nonumber\\
&\hspace{0.2cm}-\max\{\mathbb{P}(Y-\mathbb{E}[Y|X=i]<y_1|X=i),\mathbb{P}(Y-\mathbb{E}[Y|X=k]<y_2|X=k)\},0\Big\}dy_1dy_2\nonumber\\
&-\int_{{\Omega_Y}^2}\max\Big\{\min\{\mathbb{P}(Y-\mathbb{E}[Y|X=i]<y_1|X=i),\mathbb{P}(Y-\mathbb{E}[Y|X=h]<y_2|X=h)\}\nonumber\\
&\hspace{0.2cm}-\max\{\mathbb{P}(Y-\mathbb{E}[Y|X=j]<y_1|X=j),\mathbb{P}(Y-\mathbb{E}[Y|X=k]<y_2|X=k)\},0\Big\}dy_1dy_2\nonumber\\
&-\int_{{\Omega_Y}^2}\max\Big\{\min\{\mathbb{P}(Y-\mathbb{E}[Y|X=j]<y_1|X=j),\mathbb{P}(Y-\mathbb{E}[Y|X=k]<y_2|X=k)\}\nonumber\\
&\hspace{0.2cm}-\max\{\mathbb{P}(Y-\mathbb{E}[Y|X=i]<y_1|X=i),\mathbb{P}(Y-\mathbb{E}[Y|X=h]<y_2|X=h)\},0\Big\}dy_1dy_2\nonumber\\
&+\int_{{\Omega_Y}^2}\max\Big\{\min\{\mathbb{P}(Y-\mathbb{E}[Y|X=i]<y_1|X=i),\mathbb{P}(Y-\mathbb{E}[Y|X=k]<y_2|X=k)\}\nonumber\\
&\hspace{0.2cm}-\max\{\mathbb{P}(Y-\mathbb{E}[Y|X=j]<y_1|X=j),\mathbb{P}(Y-\mathbb{E}[Y|X=h]<y_2|X=h)\},0\Big\}dy_1dy_2.
\end{align}
\end{theorem}

\begin{proof}
$\mathbb{P}(Y_j-\mathbb{E}[Y_j]<y_1\leq Y_i-\mathbb{E}[Y_i],Y_h-\mathbb{E}[Y_h]<y_2\leq Y_k-\mathbb{E}[Y_k])$ are identifiable by $\max\{\min\{\mathbb{P}(Y-\mathbb{E}[Y|X=j]<y_1|X=j),\mathbb{P}(Y-\mathbb{E}[Y|X=h]<y_2|X=h)\}-\max\{\mathbb{P}(Y-\mathbb{E}[Y|X=i]<y_1|X=i),\mathbb{P}(Y-\mathbb{E}[Y|X=k]<y_2|X=k)\},0\}$ ($i,j,k,h \in \{1,\dots,R\}$) respectively by rewriting $f_Y$ as $f(x,U_Y)=Y_x-\mathbb{E}[Y_x]$ in Theorem 5.2 in \citep{Kawakami2024} under Assumptions \ref{ASEXO2} and \ref{MONO2}.
Then, we have Eq.~\eqref{eq65}.    
\end{proof}

{\bf Bounding the central product moments of causal effects.}

Assumption \ref{MONO2} may sometimes be considered implausible by researchers. Therefore, we derive bounds for the moments of causal effects without relying on Assumption \ref{MONO2}.

\begin{lemma}

Under SCM ${\cal M}$ and Assumptions \ref{ASEXO2} and \ref{exi4}, for any $i,j,k,h \in \{1,\dots,R\}$,
we have $l(y_1,y_3;i,j,k,h) \leq \mathbb{P}(Y_j-\mathbb{E}[Y_j]<y_1\leq Y_i-\mathbb{E}[Y_i],Y_h-\mathbb{E}[Y_h]<y_2\leq Y_k-\mathbb{E}[Y_k])\leq u(y_1,y_3;i,j,k,h)$, where
\begin{align}
&l(y_1,y_2;i,j,k,h)=\max\Big\{\mathbb{P}(Y-\mathbb{E}[Y|X=j]<y_1|X=j)-\mathbb{P}(Y-\mathbb{E}[Y|X=i] <y_1|X=i)\nonumber\\
&\hspace{4cm}+\mathbb{P}(Y-\mathbb{E}[Y|X=h]<y_2|X=h)-\mathbb{P}(Y-\mathbb{E}[Y|X=k]<y_2|X=k)-1,0\Big\},
\end{align}
\begin{align}
&u(y_1,y_2;i,j,k,h)=\min\Big\{\min\{\mathbb{P}(Y-\mathbb{E}[Y|X=j]<y_1|X=j),\mathbb{P}(Y-\mathbb{E}[Y|X=h]<y_2|X=h)\},\nonumber\\
&\hspace{4cm}1-\max\{\mathbb{P}(Y-\mathbb{E}[Y|X=i]<y_1|X=i),\mathbb{P}(Y-\mathbb{E}[Y|X=k]<y_2|X=k)\}\Big\}
\end{align}
for each $i,j,k,h \in \{1,\dots,R\}$ and $y_1, y_2 \in \Omega_Y$.
\end{lemma}

\begin{proof}
From Fr\'{e}chet inequalities \citep{Frechet1935,Frechet1960},
for any $i,j,k,h \in \{1,\dots,R\}$,
we have 
\begin{align}
&\mathbb{P}(Y_j-\mathbb{E}[Y_j]<y_1\leq Y_i-\mathbb{E}[Y_i],Y_h-\mathbb{E}[Y_h]<y_2\leq Y_k-\mathbb{E}[Y_k])\nonumber\\
&=\mathbb{P}(Y_j-\mathbb{E}[Y_j]<y_1,y_1\leq Y_i-\mathbb{E}[Y_i],Y_h-\mathbb{E}[Y_h]<y_2,y_2\leq Y_k-\mathbb{E}[Y_k])\nonumber\\
&\geq \max\left\{\mathbb{P}(Y_j-\mathbb{E}[Y_j]<y_1)+\mathbb{P}(y_1\leq Y_i-\mathbb{E}[Y_i])+\mathbb{P}(Y_h-\mathbb{E}[Y_h]<y_2)+\mathbb{P}(y_2\leq Y_k-\mathbb{E}[Y_k])-3,0\right\}\nonumber\\
&= \max\left\{\mathbb{P}(Y_j-\mathbb{E}[Y_j]<y_1)-\mathbb{P}(Y_i-\mathbb{E}[Y_i] <y_1)+\mathbb{P}(Y_h-\mathbb{E}[Y_h]<y_2)-\mathbb{P}(Y_k-\mathbb{E}[Y_k]<y_2)-1,0\right\}\nonumber\\
&= \max\Big\{\mathbb{P}(Y-\mathbb{E}[Y|X=j]<y_1|X=j)-\mathbb{P}(Y-\mathbb{E}[Y|X=i] <y_1|X=i)\nonumber\\
&\hspace{3cm}+\mathbb{P}(Y-\mathbb{E}[Y|X=h]<y_2|X=h)-\mathbb{P}(Y-\mathbb{E}[Y|X=k]<y_2|X=k)-1,0\big\}
\end{align}
and 
\begin{align}
&\mathbb{P}(Y_j-\mathbb{E}[Y_j]<y_1\leq Y_i-\mathbb{E}[Y_i],Y_h-\mathbb{E}[Y_h]<y_2\leq Y_k-\mathbb{E}[Y_k])\nonumber\\
&=\mathbb{P}(Y_j-\mathbb{E}[Y_j]<y_1,y_1\leq Y_i-\mathbb{E}[Y_i],Y_h-\mathbb{E}[Y_h]<y_2,y_2\leq Y_k-\mathbb{E}[Y_k])\nonumber\\
&\leq \min\Big\{\min\{\mathbb{P}(Y_j-\mathbb{E}[Y_j]<y_1),\mathbb{P}(Y_h-\mathbb{E}[Y_h]<y_2)\},1-\max\{\mathbb{P}(Y_i-\mathbb{E}[Y_i] <y_1),\mathbb{P}(Y_k-\mathbb{E}[Y_k]<y_2)\}\Big\}\nonumber\\
&=\min\Big\{\min\{\mathbb{P}(Y-\mathbb{E}[Y|X=j]<y_1|X=j),\mathbb{P}(Y-\mathbb{E}[Y|X=h]<y_2|X=h)\},\nonumber\\
&\hspace{4cm}1-\max\{\mathbb{P}(Y-\mathbb{E}[Y|X=i]<y_1|X=i),\mathbb{P}(Y-\mathbb{E}[Y|X=k]<y_2|X=k)\}\Big\}.
\end{align}
\end{proof}

\begin{theorem}[Bounds of the central product moment of causal effects]
Under SCM ${\cal M}$ and Assumptions \ref{ASEXO2} and \ref{exi4}, for any $i,j,k,h \in \{1,\dots,R\}$, we have $\bar{\sigma}_L(i,j;k,h) \leq \overline{\rho}_{i,j;k,h}\leq \bar{\sigma}_U(i,j;k,h)$, where
\begin{align}
\label{eq191}
&\bar{\sigma}_L(i,j;k,h)=\int_{{\Omega_Y}^2}\bar{l}(y_1,y_2;i,j,k,h)dy_1dy_2-\int_{{\Omega_Y}^2}\bar{u}(y_1,y_2;j,i,k,h)dy_1dy_2\nonumber\\
&\hspace{3cm}-\int_{{\Omega_Y}^2}\bar{u}(y_1,y_2;i,j,h,k)dy_1dy_2+\int_{{\Omega_Y}^2}\bar{l}(y_1,y_2;j,i,h,k)dy_1dy_2,
\end{align}
\begin{align}
\label{eq192}
&\bar{\sigma}_U(i,j;k,h)=\int_{{\Omega_Y}^2}\bar{u}(y_1,y_2;i,j,k,h)dy_1dy_2-\int_{{\Omega_Y}^2}\bar{l}(y_1,y_2;j,i,k,h)dy_1dy_2\nonumber\\
&\hspace{3cm}-\int_{{\Omega_Y}^2}\bar{l}(y_1,y_2;i,j,h,k)dy_1dy_2+\int_{{\Omega_Y}^2}\bar{u}(y_1,y_2;j,i,h,k)dy_1dy_2.
\end{align}
\end{theorem}

\begin{proof}
For any $i,j,k,h \in \{1,\dots,R\}$,
we have
\begin{align}
&\overline{\rho}_{i,j;k,h}=\int_{\Omega_Y}\int_{\Omega_Y} \mathbb{P}(Y_j-\mathbb{E}[Y_j]<y_1\leq Y_i-\mathbb{E}[Y_i],Y_h-\mathbb{E}[Y_h]<y_2\leq Y_k-\mathbb{E}[Y_k])dy_1 dy_2\nonumber\\
&-\int_{\Omega_Y}\int_{\Omega_Y} \mathbb{P}(Y_i-\mathbb{E}[Y_i]<y_1\leq Y_j-\mathbb{E}[Y_j],Y_h-\mathbb{E}[Y_h]<y_2\leq Y_k-\mathbb{E}[Y_k])dy_1 dy_2\nonumber\\
&-\int_{\Omega_Y}\int_{\Omega_Y} \mathbb{P}(Y_j-\mathbb{E}[Y_j]<y_1\leq Y_i-\mathbb{E}[Y_i],Y_k-\mathbb{E}[Y_k]<y_2\leq Y_h-\mathbb{E}[Y_h])dy_1 dy_2\nonumber\\
&+\int_{\Omega_Y}\int_{\Omega_Y} \mathbb{P}(Y_i-\mathbb{E}[Y_i]<y_1\leq Y_j-\mathbb{E}[Y_j],Y_k-\mathbb{E}[Y_k]<y_2\leq Y_h-\mathbb{E}[Y_h])dy_1 dy_2\nonumber\\
&\geq \int_{{\Omega_Y}^2}\bar{l}(y_1,y_2;i,j,k,h)dy_1dy_2-\int_{{\Omega_Y}^2}\bar{u}(y_1,y_2;j,i,k,h)dy_1dy_2\nonumber\\
&\hspace{3cm}-\int_{{\Omega_Y}^2}\bar{u}(y_1,y_2;i,j,h,k)dy_1dy_2+\int_{{\Omega_Y}^2}\bar{l}(y_1,y_2;j,i,h,k)dy_1dy_2
\end{align}
and
\begin{align}
&\overline{\rho}_{i,j;k,h}=\int_{\Omega_Y}\int_{\Omega_Y} \mathbb{P}(Y_j-\mathbb{E}[Y_j]<y_1\leq Y_i-\mathbb{E}[Y_i],Y_h-\mathbb{E}[Y_h]<y_2\leq Y_k-\mathbb{E}[Y_k])dy_1 dy_2\nonumber\\
&-\int_{\Omega_Y}\int_{\Omega_Y} \mathbb{P}(Y_i-\mathbb{E}[Y_i]<y_1\leq Y_j-\mathbb{E}[Y_j],Y_h-\mathbb{E}[Y_h]<y_2\leq Y_k-\mathbb{E}[Y_k])dy_1 dy_2\nonumber\\
&-\int_{\Omega_Y}\int_{\Omega_Y} \mathbb{P}(Y_j-\mathbb{E}[Y_j]<y_1\leq Y_i-\mathbb{E}[Y_i],Y_k-\mathbb{E}[Y_k]<y_2\leq Y_h-\mathbb{E}[Y_h])dy_1 dy_2\nonumber\\
&+\int_{\Omega_Y}\int_{\Omega_Y} \mathbb{P}(Y_i-\mathbb{E}[Y_i]<y_1\leq Y_j-\mathbb{E}[Y_j],Y_k-\mathbb{E}[Y_k]<y_2\leq Y_h-\mathbb{E}[Y_h])dy_1 dy_2\nonumber\\
&\leq \int_{{\Omega_Y}^2}\bar{u}(y_1,y_2;i,j,k,h)dy_1dy_2-\int_{{\Omega_Y}^2}\bar{l}(y_1,y_2;j,i,k,h)dy_1dy_2\nonumber\\
&\hspace{3cm}-\int_{{\Omega_Y}^2}\bar{l}(y_1,y_2;i,j,h,k)dy_1dy_2+\int_{{\Omega_Y}^2}\bar{u}(y_1,y_2;j,i,h,k)dy_1dy_2.
\end{align}
\end{proof}

Denoting the bound of the moments of causal effects $Y_i-Y_j$ by $\bar{\sigma}_L^{(m)}(i,j)$ and $\bar{\sigma}_L^{(m)}(i,j)$, i.e., $\bar{\sigma}_U^{(m)}(i,j) \leq \mathbb{E}[\{(Y_i-Y_j)-(\mathbb{E}[Y_i]-\mathbb{E}[Y_j])\}^m] \leq \bar{\sigma}_U^{(m)}(i,j)$,
we can calculate the bounds of the correlation as follows: 
\begin{align}
&\frac{\bar{\sigma}_L(i,j;k,h)}{\sqrt{\bar{\sigma}_U^{(2)}(i,j)}\sqrt{\bar{\sigma}_U^{(2)}(k,h)}}\mathbb{I}(\bar{\sigma}_L(i,j;k,h)\geq 0)+\frac{\bar{\sigma}_L(i,j;k,h)}{\sqrt{\bar{\sigma}_L^{(2)}(i,j)}\sqrt{\bar{\sigma}_L^{(2)}(k,h)}}\mathbb{I}(\bar{\sigma}_L(i,j;k,h)< 0)\nonumber\\
&\hspace{0.5cm}\leq \overline{\tau}_{i,j;k,h} \leq \frac{\bar{\sigma}_U(i,j;k,h)}{\sqrt{\bar{\sigma}_L^{(2)}(i,j)}\sqrt{\bar{\sigma}_L^{(2)}(k,h)}}\mathbb{I}(\bar{\sigma}_U(i,j;k,h)\geq 0)+\frac{\bar{\sigma}_U(i,j;k,h)}{\sqrt{\bar{\sigma}_U^{(2)}(i,j)}\sqrt{\bar{\sigma}_U^{(2)}(k,h)}}\mathbb{I}(\bar{\sigma}_U(i,j;k,h)< 0).
\end{align}

\section{Conditional Moments of Causal Effects}
\label{appD}

We consider the following SCM, ${\cal M}$:
\begin{equation}
\vspace{0.1cm}
\begin{gathered}
Y:=f_Y(X,W,U_Y),\ \ X:=f_X(W,U_X),\ \ W:=f_W(U_W)
\end{gathered}
\end{equation}
where $U_Y$, $U_X$, and $U_W$ are latent exogenous variables.
$Y$ is a continuous variable, and $X$ and $W$ are binary, discrete, or continuous variables.
$W$ are the subject's covariates.
They also provided the identification assumptions of the PNS for continuous treatment and outcome as below:
\begin{assumption}[Conditional exogeneity]
\label{CASEXO2}
$Y_x\indep X|W$ for all $x \in \Omega_X$.
\end{assumption}
If $X$ is binary, i.e., $\Omega_X=\{0,1\}$, Assumption \ref{CASEXO2} reduces to $Y_1 \indep X|W$ and $Y_0 \indep X|W$.
\begin{assumption}[Conditional monotonicity over $f_Y$]
\label{CMONO2}
{The function $f_Y(x,w,U_Y)$ is either (i) monotonic increasing on $U_Y$
with $\leq$
for all $x \in \Omega_X$ and $w \in \Omega_W$ almost surely w.r.t. $\mathbb{P}_{U_Y}$, or (ii) monotonic decreasing on $U_Y$
for all $x \in \Omega_X$ and $w \in \Omega_W$
almost surely w.r.t. $\mathbb{P}_{U_Y}$.} 
\end{assumption}
\citet{Kawakami2024} provided another Assumption (conditional monotonicity over $Y_x$), which is equivalent to Assumption \ref{CMONO2} under Assumption 3.6 in \citep{Kawakami2024}.

We provide similar definitions of the conditional moments of causal effects with the subjects' covariates $W$.
\begin{definition}[The conditional moment of causal effects]
For each $m\geq 1$,
the conditional $m$-th moment of $Y_1-Y_0$ is defined as 
\begin{equation}
\mathbb{E}\left[(Y_1-Y_0)^m\Big|W=w\right].
\end{equation}
\end{definition}

Under SCM ${\cal M}$ and Assumptions \ref{exi1}, \ref{CASEXO2}, and \ref{CMONO2}, given $m\geq 1$,
it is identified by Eq.~\eqref{eq10} conditioned on $W=w$.
Under SCM ${\cal M}$ and Assumptions \ref{exi1} and \ref{CASEXO2}, given $m\geq 1$,
it is bounded by Eqs.~\eqref{eq13} $\sim$ \eqref{eq16} conditioned on $W=w$.

\begin{definition}[The conditional central moment of causal effects]
For each $m\geq 1$,
we define the conditional $m$-th central moment of $Y_1-Y_0$ as
\begin{equation}
\mathbb{E}\left[\Big\{(Y_1-Y_0)-(\mathbb{E}[Y_1]-\mathbb{E}[Y_0])\Big\}^m\Bigg|W=w\right].
\end{equation}
\end{definition}

Under SCM ${\cal M}$ and Assumptions \ref{exi3}, \ref{CASEXO2}, and \ref{CMONO2}, given $m\geq 1$,
it is identified by Eq.~\eqref{eq91} conditioned on $W=w$.
Under SCM ${\cal M}$ and Assumptions \ref{exi3} and \ref{CASEXO2}, given $m\geq 1$,
it is bounded by Eqs.~\eqref{eq132} $\sim$ \eqref{eq135} conditioned on $W=w$.

\begin{definition}[The conditional product moment of causal effects]
We define the conditional product moment of causal effects as
\begin{equation}
\mathbb{E}\left[(Y_i-Y_j)(Y_k-Y_h)\Big|W=w\right].
\end{equation}
\end{definition}

Under SCM ${\cal M}$ and Assumptions \ref{exi2}, \ref{CASEXO2}, and \ref{CMONO2}, given $m\geq 1$,
it is identified by Eq.~\eqref{eq22} conditioned on $W=w$.
Under SCM ${\cal M}$ and Assumptions \ref{exi2} and \ref{CASEXO2}, given $m\geq 1$,
it is bounded by Eqs.~\eqref{eq25} $\sim$ \eqref{eq26} conditioned on $W=w$.

\begin{definition}[The conditional covariance of causal effects]
We define the conditional covariance (central product moment of causal effects) as
\begin{equation}
\begin{aligned}
&\mathbb{E}\left[\Big\{(Y_i-Y_j)-(\mathbb{E}[Y_i]-\mathbb{E}[Y_j])\Big\}\Big\{(Y_k-Y_h)-(\mathbb{E}[Y_k]-\mathbb{E}[Y_h])\Big\}\Bigg|W=w\right].
\end{aligned}
\end{equation}
\end{definition}

Under SCM ${\cal M}$ and Assumptions \ref{exi4}, \ref{CASEXO2}, and \ref{CMONO2}, given $m\geq 1$,
it is identified by Eq.~\eqref{eq65} conditioned on $W=w$.
Under SCM ${\cal M}$ and Assumptions \ref{exi4} and \ref{CASEXO2}, given $m\geq 1$,
it is bounded by Eqs.~\eqref{eq191} $\sim$ \eqref{eq192} conditioned on $W=w$.

\section{Consistency of Estimators}
\label{appCon}
In this appendix, we provide the details and consistency of all estimators in the body of the paper.

{\bf Details of the estimators in the body of the paper.}
The empirical CDFs and expectations are given by
\begin{gather}
\hat{\mathbb{P}}(Y<y|X=x)=\frac{\sum_{i=1}^N\mathbb{I}(Y_i<y,X_i=x)}{\sum_{i=1}^N\mathbb{I}(X_i=x)},\, \, \, \, \,
\hat{\mathbb{E}}[Y|X=x]=\frac{\sum_{i=1}^NY_i\mathbb{I}(X_i=x)}{\sum_{i=1}^N\mathbb{I}(X_i=x)}
\end{gather}
for any $x \in \Omega_X$ and $y \in \Omega_Y$.
We generate $\{y^1_{k^1}\}_{k^1=1}^{N_1}, \{y^2_{k^2}\}_{k^2=1}^{N_2}, \dots, \{y^m_{k^m}\}_{k^m=1}^{N_m}$  by i.i.d. sampling from a uniform distribution $U[\text{min}(Y),\text{max}(Y)]$ for Monte Carlo integration.

For $m=1,\dots$, the estimators $\hat{\sigma}^{(m)}$ and $\hat{\bar{\sigma}}^{(m)}$ are 
\begin{align}
&\hat{\sigma}^{(m)}=\frac{(b-a)^m}{N_1\dots N_m}\sum_{k^1=1}^{N_1}\dots\sum_{k^m=1}^{N_m}
\max\Big\{\min_{p=1,\dots,m}\{\hat{\mathbb{P}}(Y<{y^p_{k^p}}|X=0)\}-\max_{p=1,\dots,m}\{\hat{\mathbb{P}}(Y<{y^p_{k^p}}|X=1)\},0\Big\}\nonumber\\
&+(-1)^m\frac{(b-a)^m}{N_1\dots N_m}\sum_{k^1=1}^{N_1}\dots\sum_{k^m=1}^{N_m}
\max\Big\{\min_{p=1,\dots,m}\{\hat{\mathbb{P}}(Y<{y^p_{k^p}}|X=1)\}-\max_{p=1,\dots,m}\{\hat{\mathbb{P}}(Y<{y^p_{k^p}}|X=0)\},0\Big\},\nonumber\\
&\hat{\bar{\sigma}}^{(m)}=\frac{(b-a)^m}{N_1\dots N_m}\sum_{k^1=1}^{N_1}\dots\sum_{k^m=1}^{N_m}
\max\Big\{\min_{p=1,\dots,m}\{\hat{\mathbb{P}}(Y-\hat{\mathbb{E}}[Y|X=0]<{y^p_{k^p}}|X=0)\}\nonumber\\
&\hspace{8cm}-\max_{p=1,\dots,m}\{\hat{\mathbb{P}}(Y-\hat{\mathbb{E}}[Y|X=1]<{y^p_{k^p}}|X=1)\},0\Big\}\nonumber\\
&+(-1)^m\frac{(b-a)^m}{N_1\dots N_m}\sum_{k^1=1}^{N_1}\dots\sum_{k^m=1}^{N_m}
\max\Big\{\min_{p=1,\dots,m}\{\hat{\mathbb{P}}(Y-\hat{\mathbb{E}}[Y|X=1]<{y^p_{k^p}}|X=1)\}\nonumber\\
&\hspace{8cm}-\max_{p=1,\dots,m}\{\hat{\mathbb{P}}(Y-\hat{\mathbb{E}}[Y|X=0]<{y^p_{k^p}}|X=0)\},0\Big\}.
\end{align}

For $m=1,\dots$, the estimators $\hat{\sigma}_L^{(m)}$ and $\hat{\sigma}_U^{(m)}$ are given below.

(A). When $m$ is an even number,
\begin{align}
&\hat{\sigma}_L^{(m)}=\frac{(b-a)^m}{N_1\dots N_m}\sum_{k^1=1}^{N_1}\dots\sum_{k^m=1}^{N_m} \Big\{\hat{l}(y^1_{k^1},\dots,y^m_{k^m};1,0)+\hat{l}(y^1_{k^1},\dots,y^m_{k^m};0,1)\Big\},\\
&\hat{\sigma}_U^{(m)}=\frac{(b-a)^m}{N_1\dots N_m}\sum_{k^1=1}^{N_1}\dots\sum_{k^m=1}^{N_m} \Big\{\hat{u}(y^1_{k^1},\dots,y^m_{k^m};1,0)+\hat{u}(y^1_{k^1},\dots,y^m_{k^m};0,1)\Big\},
\end{align}
where 
\begin{align}
&\hat{l}(y_1,\dots,y_m;i,j)=\max\Big\{\sum_{p=1,\dots,m}\hat{\mathbb{P}}(Y<y_p|X=j)-\sum_{p=1,\dots,m}\hat{\mathbb{P}}(Y<y_p|X=i)-m+1,0\Big\},\\
&\hat{u}(y_1,\dots,y_m;i,j)=\min\Big\{\min_{p=1,\dots,m}\{\hat{\mathbb{P}}(Y<y_p|X=j)\},1-\max_{p=1,\dots,m}\{\hat{\mathbb{P}}(Y<y_p|X=i)\}\Big\}.
\end{align}

(B). When $m$ is an odd number,
\begin{align}
&\hat{\sigma}_L^{(m)}=\frac{(b-a)^m}{N_1\dots N_m}\sum_{k^1=1}^{N_1}\dots\sum_{k^m=1}^{N_m} \Big\{\hat{l}(y^1_{k^1},\dots,y^m_{k^m};1,0)-\hat{u}(y^1_{k^1},\dots,y^m_{k^m};0,1)\Big\},\\
&\hat{\sigma}_U^{(m)}=\frac{(b-a)^m}{N_1\dots N_m}\sum_{k^1=1}^{N_1}\dots\sum_{k^m=1}^{N_m} \Big\{\hat{u}(y^1_{k^1},\dots,y^m_{k^m};1,0)-\hat{l}(y^1_{k^1},\dots,y^m_{k^m};0,1)\Big\}.
\end{align}

For $m=1,\dots$, the estimators $\hat{\bar{\sigma}}_L^{(m)}$ and $\hat{\bar{\sigma}}_U^{(m)}$ are given below.

(A). When $m$ is an even number,
\begin{align}
&\hat{\bar{\sigma}}_L^{(m)}=\frac{(b-a)^m}{N_1\dots N_m}\sum_{k^1=1}^{N_1}\dots\sum_{k^m=1}^{N_m} \Big\{\hat{\bar{l}}(y^1_{k^1},\dots,y^m_{k^m};1,0)+\hat{\bar{l}}(y^1_{k^1},\dots,y^m_{k^m};0,1)\Big\},\\
&\hat{\bar{\sigma}}_U^{(m)}=\frac{(b-a)^m}{N_1\dots N_m}\sum_{k^1=1}^{N_1}\dots\sum_{k^m=1}^{N_m} \Big\{\hat{\bar{u}}(y^1_{k^1},\dots,y^m_{k^m};1,0)+\hat{\bar{u}}(y^1_{k^1},\dots,y^m_{k^m};0,1)\Big\},
\end{align}
where 
\begin{align}
&\hat{\bar{l}}(y_1,\dots,y_m;i,j)=\max\Big\{\sum_{p=1,\dots,m}\hat{\mathbb{P}}(Y-\hat{\mathbb{E}}[Y|X=j]<y_p|X=j)\nonumber\\
&\hspace{6cm}-\sum_{p=1,\dots,m}\hat{\mathbb{P}}(Y-\hat{\mathbb{E}}[Y|X=i]<y_p|X=i)-m+1,0\Big\},\\
&\hat{\bar{u}}(y_1,\dots,y_m;i,j)=\min\Big\{\min_{p=1,\dots,m}\{\hat{\mathbb{P}}(Y-\hat{\mathbb{E}}[Y|X=j]<y_p|X=j)\},\nonumber\\
&\hspace{6cm}1-\max_{p=1,\dots,m}\{\hat{\mathbb{P}}(Y-\hat{\mathbb{E}}[Y|X=i]<y_p|X=i)\}\Big\}.
\end{align}

(B). When $m$ is an odd number,
\begin{align}
&\hat{\bar{\sigma}}_L^{(m)}=\frac{(b-a)^m}{N_1\dots N_m}\sum_{k^1=1}^{N_1}\dots\sum_{k^m=1}^{N_m} \Big\{\hat{\bar{l}}(y^1_{k^1},\dots,y^m_{k^m};1,0)-\hat{\bar{u}}(y^1_{k^1},\dots,y^m_{k^m};0,1)\Big\},\\
&\hat{\bar{\sigma}}_U^{(m)}=\frac{(b-a)^m}{N_1\dots N_m}\sum_{k^1=1}^{N_1}\dots\sum_{k^m=1}^{N_m} \Big\{\hat{\bar{u}}(y^1_{k^1},\dots,y^m_{k^m};1,0)-\hat{\bar{l}}(y^1_{k^1},\dots,y^m_{k^m};0,1)\Big\}.
\end{align}

The estimators $\hat{\sigma}(i,j;k,h)$ and $\hat{\bar{\sigma}}(i,j;k,h)$ are 
\begin{align}
\hat{\sigma}(i,j;k,h)&=\frac{(b-a)^2}{N_1N_2}\sum_{k^1=1}^{N_1}\sum_{k^2=1}^{N_2}\max\Big\{\min\{\hat{\mathbb{P}}(Y<y^1_{k^1}|X=j),\hat{\mathbb{P}}(Y<y^2_{k^2}|X=h)\}\nonumber\\
&\hspace{5cm}-\max\{\hat{\mathbb{P}}(Y<y^1_{k^1}|X=i),\hat{\mathbb{P}}(Y<y^2_{k^2}|X=k)\},0\Big\}\nonumber\\
&-\frac{(b-a)^2}{N_1N_2}\sum_{k^1=1}^{N_1}\sum_{k^2=1}^{N_2}\max\Big\{\min\{\hat{\mathbb{P}}(Y<y^1_{k^1}|X=i),\hat{\mathbb{P}}(Y<y^2_{k^2}|X=h)\}\nonumber\\
&\hspace{5cm}-\max\{\hat{\mathbb{P}}(Y<y^1_{k^1}|X=j),\hat{\mathbb{P}}(Y<y^2_{k^2}|X=k)\},0\Big\}\nonumber\\
&-\frac{(b-a)^2}{N_1N_2}\sum_{k^1=1}^{N_1}\sum_{k^2=1}^{N_2}\max\Big\{\min\{\hat{\mathbb{P}}(Y<y^1_{k^1}|X=j),\hat{\mathbb{P}}(Y<y^2_{k^2}|X=k)\}\nonumber\\
&\hspace{5cm}-\max\{\hat{\mathbb{P}}(Y<y^1_{k^1}|X=i),\hat{\mathbb{P}}(Y<y^2_{k^2}|X=h)\},0\Big\}\nonumber\\
&+\frac{(b-a)^2}{N_1N_2}\sum_{k^1=1}^{N_1}\sum_{k^2=1}^{N_2}\max\Big\{\min\{\hat{\mathbb{P}}(Y<y^1_{k^1}|X=i),\hat{\mathbb{P}}(Y<y^2_{k^2}|X=k)\}\nonumber\\
&\hspace{5cm}-\max\{\hat{\mathbb{P}}(Y<y^1_{k^1}|X=j),\hat{\mathbb{P}}(Y<y^2_{k^2}|X=h)\},0\Big\},
\end{align}
\begin{align}
\hat{\bar{\sigma}}(i,j;k,h)
&=\frac{(b-a)^2}{N_1N_2}\sum_{k^1=1}^{N_1}\sum_{k^2=1}^{N_2}\max\Big\{\min\{\hat{\mathbb{P}}(Y-\hat{\mathbb{E}}[Y|X=j]<y^1_{k^1}|X=j),\hat{\mathbb{P}}(Y-\hat{\mathbb{E}}[Y|X=h]<y^2_{k^2}|X=h)\}\nonumber\\
&\hspace{3cm}-\max\{\hat{\mathbb{P}}(Y-\hat{\mathbb{E}}[Y|X=i]<y^1_{k^1}|X=i),\hat{\mathbb{P}}(Y-\hat{\mathbb{E}}[Y|X=k]<y^2_{k^2}|X=k)\},0\Big\}\nonumber\\
&-\frac{(b-a)^2}{N_1N_2}\sum_{k^1=1}^{N_1}\sum_{k^2=1}^{N_2}\max\Big\{\min\{\hat{\mathbb{P}}(Y-\hat{\mathbb{E}}[Y|X=i]<y^1_{k^1}|X=i),\hat{\mathbb{P}}(Y-\hat{\mathbb{E}}[Y|X=h]<y^2_{k^2}|X=h)\}\nonumber\\
&\hspace{3cm}-\max\{\hat{\mathbb{P}}(Y-\hat{\mathbb{E}}[Y|X=j]<y^1_{k^1}|X=j),\hat{\mathbb{P}}(Y-\hat{\mathbb{E}}[Y|X=k]<y^2_{k^2}|X=k)\},0\Big\}\nonumber\\
&-\frac{(b-a)^2}{N_1N_2}\sum_{k^1=1}^{N_1}\sum_{k^2=1}^{N_2}\max\Big\{\min\{\hat{\mathbb{P}}(Y-\hat{\mathbb{E}}[Y|X=j]<y^1_{k^1}|X=j),\hat{\mathbb{P}}(Y-\hat{\mathbb{E}}[Y|X=k]<y^2_{k^2}|X=k)\}\nonumber\\
&\hspace{3cm}-\max\{\hat{\mathbb{P}}(Y-\hat{\mathbb{E}}[Y|X=i]<y^1_{k^1}|X=i),\hat{\mathbb{P}}(Y-\hat{\mathbb{E}}[Y|X=h]<y^2_{k^2}|X=h)\},0\Big\}\nonumber\\
&+\frac{(b-a)^2}{N_1N_2}\sum_{k^1=1}^{N_1}\sum_{k^2=1}^{N_2}\max\Big\{\min\{\hat{\mathbb{P}}(Y-\hat{\mathbb{E}}[Y|X=i]<y^1_{k^1}|X=i),\hat{\mathbb{P}}(Y-\hat{\mathbb{E}}[Y|X=k]<y^2_{k^2}|X=k)\}\nonumber\\
&\hspace{3cm}-\max\{\hat{\mathbb{P}}(Y-\hat{\mathbb{E}}[Y|X=j]<y^1_{k^1}|X=j),\hat{\mathbb{P}}(Y-\hat{\mathbb{E}}[Y|X=h]<y^2_{k^2}|X=h)\},0\Big\}.
\end{align}

The estimators $\hat{\sigma}_L(i,j;k,h)$ and $\hat{\sigma}_U(i,j;k,h)$ are 
\begin{align}
\hat{\sigma}_L(i,j;k,h)
&=\frac{(b-a)^2}{N_1N_2}\sum_{k^1=1}^{N_1}\sum_{k^2=1}^{N_2}\hat{l}(y^1_{k^1},y^2_{k^2};i,j,k,h)-\frac{(b-a)^2}{N_1N_2}\sum_{k^1=1}^{N_1}\sum_{k^2=1}^{N_2}\hat{u}(y^1_{k^1},y^2_{k^2};j,i,k,h)\nonumber\\
&\hspace{1cm}-\frac{(b-a)^2}{N_1N_2}\sum_{k^1=1}^{N_1}\sum_{k^2=1}^{N_2}\hat{u}(y^1_{k^1},y^2_{k^2};i,j,h,k)+\frac{(b-a)^2}{N_1N_2}\sum_{k^1=1}^{N_1}\sum_{k^2=1}^{N_2}\hat{l}(y^1_{k^1},y^2_{k^2};j,i,h,k),\\
\hat{\sigma}_U(i,j;k,h)
&=\frac{(b-a)^2}{N_1N_2}\sum_{k^1=1}^{N_1}\sum_{k^2=1}^{N_2}\hat{u}(y^1_{k^1},y^2_{k^2};i,j,k,h)-\frac{(b-a)^2}{N_1N_2}\sum_{k^1=1}^{N_1}\sum_{k^2=1}^{N_2}\hat{l}(y^1_{k^1},y^2_{k^2};j,i,k,h)\nonumber\\
&\hspace{1cm}-\frac{(b-a)^2}{N_1N_2}\sum_{k^1=1}^{N_1}\sum_{k^2=1}^{N_2}\hat{l}(y^1_{k^1},y^2_{k^2};i,j,h,k)+\frac{(b-a)^2}{N_1N_2}\sum_{k^1=1}^{N_1}\sum_{k^2=1}^{N_2}\hat{u}(y^1_{k^1},y^2_{k^2};j,i,h,k),
\end{align}
where
\begin{align}
&\hat{l}(y_1,y_2;i,j,k,h)\nonumber\\
&=\max\Big\{\hat{\mathbb{P}}(Y<y_1|X=j)-\hat{\mathbb{P}}(Y <y_1|X=i)+\hat{\mathbb{P}}(Y<y_2|X=h)-\hat{\mathbb{P}}(Y<y_2|X=k)-1,0\Big\},\\
&\hat{u}(y_1,y_2;i,j,k,h)\nonumber\\
&=\min\Big\{\min\{\hat{\mathbb{P}}(Y<y_1|X=j),\hat{\mathbb{P}}(Y<y_2|X=h)\},1-\max\{\hat{\mathbb{P}}(Y<y_1|X=i),\hat{\mathbb{P}}(Y<y_2|X=k)\}\Big\}.
\end{align}

The estimators $\hat{\bar{\sigma}}_L(i,j;k,h)$ and $\hat{\bar{\sigma}}_U(i,j;k,h)$ are 
\begin{align}
\hat{\bar{\sigma}}_L(i,j;k,h)
&=\frac{(b-a)^2}{N_1N_2}\sum_{k^1=1}^{N_1}\sum_{k^2=1}^{N_2}\hat{\bar{l}}(y^1_{k^1},y^2_{k^2};i,j,k,h)-\frac{(b-a)^2}{N_1N_2}\sum_{k^1=1}^{N_1}\sum_{k^2=1}^{N_2}\hat{\bar{u}}(y^1_{k^1},y^2_{k^2};j,i,k,h)\nonumber\\
&\hspace{1cm}-\frac{(b-a)^2}{N_1N_2}\sum_{k^1=1}^{N_1}\sum_{k^2=1}^{N_2}\hat{\bar{u}}(y^1_{k^1},y^2_{k^2};i,j,h,k)+\frac{(b-a)^2}{N_1N_2}\sum_{k^1=1}^{N_1}\sum_{k^2=1}^{N_2}\hat{\bar{l}}(y^1_{k^1},y^2_{k^2};j,i,h,k),\\
\hat{\bar{\sigma}}_U(i,j;k,h)
&=\frac{(b-a)^2}{N_1N_2}\sum_{k^1=1}^{N_1}\sum_{k^2=1}^{N_2}\hat{\bar{u}}(y^1_{k^1},y^2_{k^2};i,j,k,h)-\frac{(b-a)^2}{N_1N_2}\sum_{k^1=1}^{N_1}\sum_{k^2=1}^{N_2}\hat{\bar{l}}(y^1_{k^1},y^2_{k^2};j,i,k,h)\nonumber\\
&\hspace{1cm}-\frac{(b-a)^2}{N_1N_2}\sum_{k^1=1}^{N_1}\sum_{k^2=1}^{N_2}\hat{\bar{l}}(y^1_{k^1},y^2_{k^2};i,j,h,k)+\frac{(b-a)^2}{N_1N_2}\sum_{k^1=1}^{N_1}\sum_{k^2=1}^{N_2}\hat{\bar{u}}(y^1_{k^1},y^2_{k^2};j,i,h,k),
\end{align}
where
\begin{align}
\hat{\bar{l}}(y_1,y_2;i,j,k,h)
&=\max\Big\{\hat{\mathbb{P}}(Y-\hat{\mathbb{E}}[Y|X=j]<y_1|X=j)-\hat{\mathbb{P}}(Y-\hat{\mathbb{E}}[Y|X=i]<y_1|X=i)\nonumber\\
&\hspace{1cm}+\hat{\mathbb{P}}(Y-\hat{\mathbb{E}}[Y|X=h]<y_2|X=h)-\hat{\mathbb{P}}(Y-\hat{\mathbb{E}}[Y|X=k]<y_2|X=k)-1,0\Big\},\\
\hat{\bar{u}}(y_1,y_2;i,j,k,h)
&=\min\Big\{\min\{\hat{\mathbb{P}}(Y-\hat{\mathbb{E}}[Y|X=j]<y_1|X=j),\hat{\mathbb{P}}(Y-\hat{\mathbb{E}}[Y|X=h]<y_2|X=h)\},\nonumber\\
&\hspace{1cm}1-\max\{\hat{\mathbb{P}}(Y-\hat{\mathbb{E}}[Y|X=i]<y_1|X=i),\hat{\mathbb{P}}(Y-\hat{\mathbb{E}}[Y|X=k]<y_2|X=k)\}\Big\}.
\end{align}

{\bf Consistency of the estimators for the moments of causal effects.}
First, the empirical CDFs and expectations follow $\displaystyle O_p\left(1/\sqrt{N}\right)$ for any $x \in \Omega_X$ and $y \in \Omega_Y$.
From the delta method \citep{Doob1935,Fang2018},
both $\max\{\min_{p=1,\dots,m}\{\hat{\mathbb{P}}(Y<y_p|X=0)\}-\max_{p=1,\dots,m}\{\hat{\mathbb{P}}(Y<y_p|X=1)\},0\}$ and $\max\{\min_{p=1,\dots,m}\{\hat{\mathbb{P}}(Y<y_p|X=1)\}-\max_{p=1,\dots,m}\{\hat{\mathbb{P}}(Y<y_p|X=0)\},0\}$ follow $ O_p\left(1/\sqrt{N^m}\right)$, {almost surely} w.r.t. $\Omega_Y^m$.
For any $m=1,\dots$, we have
\begin{align}
&\Big|\hat{\sigma}^{(m)}-\sigma^{(m)}\Big|\nonumber\\
&=\Bigg|\Bigg\{\frac{(b-a)^m}{N_1\dots N_m}\sum_{k^1=1}^{N_1}\dots\sum_{k^m=1}^{N_m}
\max\Big\{\min_{p=1,\dots,m}\{\hat{\mathbb{P}}(Y<{y^p_{k^p}}|X=0)\}-\max_{p=1,\dots,m}\{\hat{\mathbb{P}}(Y<{y^p_{k^p}}|X=1)\},0\Big\}\nonumber\\
&+(-1)^m\frac{(b-a)^m}{N_1\dots N_m}\sum_{k^1=1}^{N_1}\dots\sum_{k^m=1}^{N_m}
\max\Big\{\min_{p=1,\dots,m}\{\hat{\mathbb{P}}(Y<{y^p_{k^p}}|X=1)\}-\max_{p=1,\dots,m}\{\hat{\mathbb{P}}(Y<{y^p_{k^p}}|X=0)\},0\Big\}\Bigg\}\nonumber\\
&-\Bigg\{\sigma^{(m)}=\int_{{\Omega_Y}^m} \max\Big\{\min_{p=1,\dots,m}\{\mathbb{P}(Y<y_p|X=0)\}-\max_{p=1,\dots,m}\{\mathbb{P}(Y<y_p|X=1)\},0\Big\}dy_1\dots dy_m\nonumber\\
&\hspace{1cm}+(-1)^m\int_{{\Omega_Y}^m} \max\Big\{\min_{p=1,\dots,m}\{\mathbb{P}(Y<y_p|X=1)\}-\max_{p=1,\dots,m}\{\mathbb{P}(Y<y_p|X=0)\},0\Big\}dy_1\dots dy_m\Bigg\}\Bigg|\nonumber\\
&=\Bigg|\Bigg\{\frac{(b-a)^m}{N_1\dots N_m}\sum_{k^1=1}^{N_1}\dots\sum_{k^m=1}^{N_m}
\max\Big\{\min_{p=1,\dots,m}\{\hat{\mathbb{P}}(Y<{y^p_{k^p}}|X=0)\}-\max_{p=1,\dots,m}\{\hat{\mathbb{P}}(Y<{y^p_{k^p}}|X=1)\},0\Big\}\nonumber\\
&+(-1)^m\frac{(b-a)^m}{N_1\dots N_m}\sum_{k^1=1}^{N_1}\dots\sum_{k^m=1}^{N_m}
\max\Big\{\min_{p=1,\dots,m}\{\hat{\mathbb{P}}(Y<{y^p_{k^p}}|X=1)\}-\max_{p=1,\dots,m}\{\hat{\mathbb{P}}(Y<{y^p_{k^p}}|X=0)\},0\Big\}\Bigg\}\nonumber\\
&-\Bigg\{\int_{{\Omega_Y}^m} \max\Big\{\min_{p=1,\dots,m}\{\hat{\mathbb{P}}(Y<y_p|X=0)\}-\max_{p=1,\dots,m}\{\hat{\mathbb{P}}(Y<y_p|X=1)\},0\Big\}dy_1\dots dy_m\nonumber\\
&\hspace{1cm}+(-1)^m\int_{{\Omega_Y}^m} \max\Big\{\min_{p=1,\dots,m}\{\hat{\mathbb{P}}(Y<y_p|X=1)\}-\max_{p=1,\dots,m}\{\hat{\mathbb{P}}(Y<y_p|X=0)\},0\Big\}dy_1\dots dy_m\Bigg\}\nonumber\\
&+\Bigg\{\int_{{\Omega_Y}^m} \max\Big\{\min_{p=1,\dots,m}\{\hat{\mathbb{P}}(Y<y_p|X=0)\}-\max_{p=1,\dots,m}\{\hat{\mathbb{P}}(Y<y_p|X=1)\},0\Big\}dy_1\dots dy_m\nonumber\\
&\hspace{1cm}+(-1)^m\int_{{\Omega_Y}^m} \max\Big\{\min_{p=1,\dots,m}\{\hat{\mathbb{P}}(Y<y_p|X=1)\}-\max_{p=1,\dots,m}\{\hat{\mathbb{P}}(Y<y_p|X=0)\},0\Big\}dy_1\dots dy_m\Bigg\}\nonumber\\
&-\Bigg\{\int_{{\Omega_Y}^m} \max\Big\{\min_{p=1,\dots,m}\{\mathbb{P}(Y<y_p|X=0)\}-\max_{p=1,\dots,m}\{\mathbb{P}(Y<y_p|X=1)\},0\Big\}dy_1\dots dy_m\nonumber\\
&\hspace{1cm}+(-1)^m\int_{{\Omega_Y}^m} \max\Big\{\min_{p=1,\dots,m}\{\mathbb{P}(Y<y_p|X=1)\}-\max_{p=1,\dots,m}\{\mathbb{P}(Y<y_p|X=0)\},0\Big\}dy_1\dots dy_m\Bigg\}\Bigg|\nonumber\\
&\leq\Bigg|\Bigg\{\frac{(b-a)^m}{N_1\dots N_m}\sum_{k^1=1}^{N_1}\dots\sum_{k^m=1}^{N_m}
\max\Big\{\min_{p=1,\dots,m}\{\hat{\mathbb{P}}(Y<{y^p_{k^p}}|X=0)\}-\max_{p=1,\dots,m}\{\hat{\mathbb{P}}(Y<{y^p_{k^p}}|X=1)\},0\Big\}\nonumber\\
&+(-1)^m\frac{(b-a)^m}{N_1\dots N_m}\sum_{k^1=1}^{N_1}\dots\sum_{k^m=1}^{N_m}
\max\Big\{\min_{p=1,\dots,m}\{\hat{\mathbb{P}}(Y<{y^p_{k^p}}|X=1)\}-\max_{p=1,\dots,m}\{\hat{\mathbb{P}}(Y<{y^p_{k^p}}|X=0)\},0\Big\}\Bigg\}\nonumber\\
&-\Bigg\{\int_{{\Omega_Y}^m} \max\Big\{\min_{p=1,\dots,m}\{\hat{\mathbb{P}}(Y<y_p|X=0)\}-\max_{p=1,\dots,m}\{\hat{\mathbb{P}}(Y<y_p|X=1)\},0\Big\}dy_1\dots dy_m\nonumber\\
&\hspace{1cm}+(-1)^m\int_{{\Omega_Y}^m} \max\Big\{\min_{p=1,\dots,m}\{\hat{\mathbb{P}}(Y<y_p|X=1)\}-\max_{p=1,\dots,m}\{\hat{\mathbb{P}}(Y<y_p|X=0)\},0\Big\}dy_1\dots dy_m\Bigg\}\Bigg|\nonumber\\
&+\Bigg|\Bigg\{\int_{{\Omega_Y}^m} \max\Big\{\min_{p=1,\dots,m}\{\hat{\mathbb{P}}(Y<y_p|X=0)\}-\max_{p=1,\dots,m}\{\hat{\mathbb{P}}(Y<y_p|X=1)\},0\Big\}dy_1\dots dy_m\nonumber\\
&\hspace{1cm}+(-1)^m\int_{{\Omega_Y}^m} \max\Big\{\min_{p=1,\dots,m}\{\hat{\mathbb{P}}(Y<y_p|X=1)\}-\max_{p=1,\dots,m}\{\hat{\mathbb{P}}(Y<y_p|X=0)\},0\Big\}dy_1\dots dy_m\Bigg\}\nonumber\\
&-\Bigg\{\int_{{\Omega_Y}^m} \max\Big\{\min_{p=1,\dots,m}\{\mathbb{P}(Y<y_p|X=0)\}-\max_{p=1,\dots,m}\{\mathbb{P}(Y<y_p|X=1)\},0\Big\}dy_1\dots dy_m\nonumber\\
&\hspace{1cm}+(-1)^m\int_{{\Omega_Y}^m} \max\Big\{\min_{p=1,\dots,m}\{\mathbb{P}(Y<y_p|X=1)\}-\max_{p=1,\dots,m}\{\mathbb{P}(Y<y_p|X=0)\},0\Big\}dy_1\dots dy_m\Bigg\}\Bigg|\nonumber\\
&=\Bigg|\Bigg\{\frac{(b-a)^m}{N_1\dots N_m}\sum_{k^1=1}^{N_1}\dots\sum_{k^m=1}^{N_m}
\max\Big\{\min_{p=1,\dots,m}\{\hat{\mathbb{P}}(Y<{y^p_{k^p}}|X=0)\}-\max_{p=1,\dots,m}\{\hat{\mathbb{P}}(Y<{y^p_{k^p}}|X=1)\},0\Big\}\nonumber\\
&+(-1)^m\frac{(b-a)^m}{N_1\dots N_m}\sum_{k^1=1}^{N_1}\dots\sum_{k^m=1}^{N_m}
\max\Big\{\min_{p=1,\dots,m}\{\hat{\mathbb{P}}(Y<{y^p_{k^p}}|X=1)\}-\max_{p=1,\dots,m}\{\hat{\mathbb{P}}(Y<{y^p_{k^p}}|X=0)\},0\Big\}\Bigg\}\nonumber\\
&-\Bigg\{\int_{{\Omega_Y}^m} \max\Big\{\min_{p=1,\dots,m}\{\hat{\mathbb{P}}(Y<y_p|X=0)\}-\max_{p=1,\dots,m}\{\hat{\mathbb{P}}(Y<y_p|X=1)\},0\Big\}dy_1\dots dy_m\nonumber\\
&\hspace{1cm}+(-1)^m\int_{{\Omega_Y}^m} \max\Big\{\min_{p=1,\dots,m}\{\hat{\mathbb{P}}(Y<y_p|X=1)\}-\max_{p=1,\dots,m}\{\hat{\mathbb{P}}(Y<y_p|X=0)\},0\Big\}dy_1\dots dy_m\Bigg\}\Bigg|\nonumber\\
&+\Bigg|\int_{{\Omega_Y}^m}\Bigg\{\max\Big\{\min_{p=1,\dots,m}\{\hat{\mathbb{P}}(Y<y_p|X=0)\}-\max_{p=1,\dots,m}\{\hat{\mathbb{P}}(Y<y_p|X=1)\},0\Big\}\nonumber\\
&\hspace{4cm}+(-1)^m\max\Big\{\min_{p=1,\dots,m}\{\hat{\mathbb{P}}(Y<y_p|X=1)\}-\max_{p=1,\dots,m}\{\hat{\mathbb{P}}(Y<y_p|X=0)\},0\Big\}\nonumber\\
&-\max\Big\{\min_{p=1,\dots,m}\{\mathbb{P}(Y<y_p|X=0)\}-\max_{p=1,\dots,m}\{\mathbb{P}(Y<y_p|X=1)\},0\Big\}\nonumber\\
&\hspace{2cm}-(-1)^m\max\Big\{\min_{p=1,\dots,m}\{\mathbb{P}(Y<y_p|X=1)\}-\max_{p=1,\dots,m}\{\mathbb{P}(Y<y_p|X=0)\},0\Big\}\Bigg\}dy_1\dots dy_m\Bigg|\nonumber\\
&= O_p\left(\sum_{i=1}^m\frac{1}{\sqrt{N_i}}\right)+ O_p\left(\frac{1}{\sqrt{N^m}}\right).
\end{align}
Then, under SCM ${\cal M}$ and Assumptions \ref{ASEXO2}, \ref{MONO2}, and \ref{exi1}, $\sigma^{(m)}$ follows $\displaystyle O_p\left(1/\sqrt{N^m}+\sum\nolimits_{i=1}^m1/\sqrt{N_i}\right)$.
Similarly, under SCM ${\cal M}$ and Assumptions \ref{ASEXO2}, \ref{MONO2}, and \ref{exi1}, $\displaystyle\sigma_L^{(m)}, \sigma_U^{(m)}$ follow $\displaystyle O_p\left(1/\sqrt{N^m}+\sum\nolimits_{i=1}^m1/\sqrt{N_i}\right)$.
Under SCM ${\cal M}$ and Assumptions \ref{ASEXO2}, \ref{MONO2}, and \ref{exi2}, $\displaystyle\sigma(i,j;k,h), \sigma_L(i,j;k,h), \sigma_U(i,j;k,h)$ follow $\displaystyle O_p\left(1/\sqrt{N^2}+\sum\nolimits_{i=1}^21/\sqrt{N_i}\right)$.
We can make $N_1,\dots,N_m$ as large as computational resources allow. 
Letting $N_1,\dots,N_m \rightarrow \infty$, $\displaystyle\sigma^{(m)}, \sigma_L^{(m)}, \sigma_U^{(m)}$ follow $\displaystyle O_p\left(1/\sqrt{N^m}\right)$ and $\displaystyle\sigma(i,j;k,h), \sigma_L(i,j;k,h), \sigma_U(i,j;k,h)$ follow $\displaystyle O_p\left(1/\sqrt{N^2}\right)$.

{\bf Consistency of the estimators for the central moments of causal effects.}
First, empirical CDFs and expectations follow $\displaystyle O_p\left(1/\sqrt{N}\right)$ for any $x \in \Omega_X$ and $y \in \Omega_Y$.
For the central moment, we make an additional assumption.
\begin{assumption}
\label{Lip}
$\mathbb{P}(Y<y|X=x)$ is differential in $y$ for any $x \in \Omega_X$, almost surely w.r.t. $\Omega_Y$.
\end{assumption}
Then, under SCM ${\cal M}$ and Assumptions \ref{ASEXO2}, \ref{MONO2}, \ref{exi3}, and \ref{Lip},
$\sigma^{(m)}, \sigma_L^{(m)}, \sigma_U^{(m)}$ follow $\displaystyle O_p\left(1/\sqrt{N^{2m}}+\sum\nolimits_{i=1}^m1/\sqrt{N_i}\right)$.
Under SCM ${\cal M}$ and Assumptions \ref{ASEXO2}, \ref{MONO2}, \ref{exi4}, and \ref{Lip},
$\sigma(i,j;k,h), \sigma_L(i,j;k,h), \sigma_U(i,j;k,h)$ follow $\displaystyle O_p\left(1/\sqrt{N^4}+\sum\nolimits_{i=1}^21/\sqrt{N_i}\right)$.
Letting $N_1,\dots,N_m \rightarrow \infty$, 
$\sigma^{(m)}, \sigma_L^{(m)}, \sigma_U^{(m)}$ follow $\displaystyle O_p\left(1/\sqrt{N^{2m}}\right)$ and $\sigma(i,j;k,h), \sigma_L(i,j;k,h), \sigma_U(i,j;k,h)$ follow $\displaystyle O_p\left(1/\sqrt{N^4}\right)$.

Thus, all estimators for the moments of causal effects are consistent.


\section{Additional Information about the Application to Real-World}
\label{appE}

In this section, we provide additional information about the application in the body of our paper.

{\bf Additional information about the moments of causal effects.}

We provide the bounds of variance, standard deviation, skewness, and kurtosis of causal effect $Y_2-Y_1$.
When the denominator of the estimator is zero, we replace it with $0.01$ to avoid numerical instability.
The estimated bounds of the causal effect $Y_2-Y_1$ are
\begin{center}
\textbf{Upper bound of variance}: $29.851$ (95\%CI: $[12.173,53.740]$),\\\vspace{0.1cm}
\textbf{Lower bound of variance}: $0.000$ (95\%CI: $[0.000,0.000]$),\\\vspace{0.1cm}
\textbf{Upper bound of standard deviation}: $5.463$ (95\%CI: $[3.489,7.33]$),\\\vspace{0.1cm}
\textbf{Lower bound of standard deviation}: $0.000$ (95\%CI: $[0.000,0.000]$),\\\vspace{0.1cm}
\textbf{Upper bound of skewness}: $71296.725$ (95\%CI: $[0.000,258849.124]$),\\\vspace{0.1cm}
\textbf{Lower bound of skewness}: $-13168.940$ (95\%CI: $[-64712.250,0.000]$),\\\vspace{0.1cm}
\textbf{Upper bound of kurtosis}: $6637894.937$ (95\%CI: $[0.000,35261058.567]$),\\\vspace{0.1cm}
\textbf{Lower bound of kurtosis}: $0.000$ (95\%CI: $[0.000,0.000]$).
\end{center}

We also provide the variance, skewness, and kurtosis of causal effect $Y_4-Y_2$.
The results are:
\begin{center}
\textbf{Mean}: $3.124$ (95\%CI: $[0.364,5.820]$),\\\vspace{0.1cm}
\textbf{Variance}: $3.104$ (95\%CI: $[0.297,8.610]$),\\\vspace{0.1cm}
\textbf{Skewness}: $-3.373$ (95\%CI: $[-28.287,8.821]$),\\\vspace{0.1cm}
\textbf{Kurtosis}: $18.765$ (95\%CI: $[0.000,203.983]$).
\end{center}
The estimated bounds of the causal effect $Y_4-Y_2$ are
\begin{center}
\textbf{Upper bound of variance}: $29.503$ (95\%CI: $[10.532,61.318]$),\\\vspace{0.1cm}
\textbf{Lower bound of variance}: $0.000$ (95\%CI: $[0.000,0.000]$),\\\vspace{0.1cm}
\textbf{Upper bound of skewness}: $34944.621$ (95\%CI: $[0.000,129424.512]$),\\\vspace{0.1cm}
\textbf{Lower bound of skewness}: $-79272.510$ (95\%CI: $[-268151.400,0.000]$),\\\vspace{0.1cm}
\textbf{Upper bound of kurtosis}: $3261648.987$ (95\%CI: $[0.000,22258543.128]$),\\\vspace{0.1cm}
\textbf{Lower bound of kurtosis}: $0.000$ (95\%CI: $[0.000,0.000]$).
\end{center}
We provide the variance, skewness, and kurtosis of causal effect $Y_4-Y_1$.
The results are:
\begin{center}
\textbf{Mean}: $6.551$ (95\%CI: $[4.549,9.034]$),\\\vspace{0.1cm}
\textbf{Variance}: $2.078$ (95\%CI: $[0.000,5.953]$),\\\vspace{0.1cm}
\textbf{Skewness}: $5.797$ (95\%CI: $[-24.688,65.521]$),\\\vspace{0.1cm}
\textbf{Kurtosis}: $14.338$ (95\%CI: $[0.000,242.5153]$).
\end{center}
The estimated bounds of the causal effect $Y_4-Y_1$ are
\begin{center}
\textbf{Upper bound of variance}: $25.160$ (95\%CI: $[9.034,42.910]$),\\\vspace{0.1cm}
\textbf{Lower bound of variance}: $0.000$ (95\%CI: $[0.000,0.000]$),\\\vspace{0.1cm}
\textbf{Upper bound of skewness}: $34135.710$ (95\%CI: $[0.000,113246.400]$),\\\vspace{0.1cm}
\textbf{Lower bound of skewness}: $-19251.891$ (95\%CI: $[-97068.380,0.000]$),\\\vspace{0.1cm}
\textbf{Upper bound of kurtosis}: $2291969,972$ (95\%CI: $[0.000,17630519.123]$),\\\vspace{0.1cm}
\textbf{Lower bound of kurtosis}: $0.000$ (95\%CI: $[0.000,0.000]$).
\end{center}

{\bf Additional information about the product moments of causal effects.}

We provide the bounds of covariance and correlation of causal effects $Y_2-Y_1$ and $Y_4-Y_2$.
When the denominator of the estimator is zero, it is replaced with $0.01$ to prevent numerical instability.
The estimated bounds of the covariance of causal effect $Y_2-Y_1$ and $Y_4-Y_2$ are
\begin{center}
\textbf{Upper bound of covariance}: $27.613$ (95\%CI: $[9.486,50.481]$),\\\vspace{0.1cm}
\textbf{Lower bound of covariance}: $-29.172$ (95\%CI: $[-52.263,-11.282]$),
\end{center}
and the estimated bounds of the correlation of causal effect $Y_2-Y_1$ and $Y_4-Y_2$ are
\begin{center}
\textbf{Upper bound of correlation}: $0.919$ (95\%CI: $[0.000,1.000]$),\\\vspace{0.1cm}
\textbf{Lower bound of correlation}: $-0.976$ (95\%CI: $[-1.000,-0.923]$).
\end{center}


\end{document}